\documentclass[11pt]{article}
\usepackage[dvipsnames]{xcolor}
\usepackage{authblk}
\author[1,2]{Leon Menger\,\thanks{leonmenger@nd.edu}\!\!$~^,$}
\author[1,2]{Pavel Mnev\,\thanks{pmnev@nd.edu}\!\!$~^,$}
\affil[1]{Department of Mathematics, Notre Dame University, 255 Hurley Bldg, Notre Dame, IN 46556, USA}
\affil[2]{Institut für Mathematik, Universität Zürich, Winterthurerstrasse 190,
CH-8057 Zürich, Switzerland}
\let\oldphi\varphi \let\varphi\phi \let\phi\oldphi
\let\oldepsilon\varepsilon \let\varepsilon\epsilon \let\epsilon\oldepsilon
\usepackage{geometry}
\usepackage{float}
\usepackage{amsmath}
\usepackage{mathtools}
\usepackage{amsthm}
\usepackage{amssymb}
\usepackage{thmtools}
\usepackage[bbgreekl]{mathbbol}
\usepackage[utf8]{inputenc}
\usepackage{graphicx}
\usepackage[T1]{fontenc}
\usepackage{amsfonts}
\usepackage[english]{babel}
\usepackage{csquotes}
\usepackage{fancyhdr}
\usepackage[font=small,labelfont=bf]{caption}
\usepackage{tikz-cd}
\usepackage{dsfont}
\usepackage{enumitem}
\usepackage{stmaryrd}
\usepackage{tensor}
\usepackage[inkscapelatex=false]{svg}
\usepackage[normalem]{ulem}
\numberwithin{equation}{section}
\usepackage[backend=biber, giveninits=true, style=alphabetic, sorting=nyt,isbn=false, doi=false, maxalphanames=5,url=false, maxbibnames=99]{biblatex}
\usepackage{comment}
\usepackage{multirow}
\usepackage[hypertexnames=false]{hyperref}
\usepackage[capitalise]{cleveref}
\mathtoolsset{showonlyrefs = false}
\newcommand{\CB}{\ensuremath{\mathds{C}}}

\newcommand{\LB}{\ensuremath{\mathds{L}}}

\newcommand{\NB}{\ensuremath{\mathds{N}}}

\newcommand{\PB}{\ensuremath{\mathds{P}}}

\newcommand{\RB}{\ensuremath{\mathds{R}}}

\newcommand{\TB}{\ensuremath{\mathds{T}}}

\newcommand{\ZB}{\ensuremath{\mathds{Z}}}
\newcommand{\tb}{\ensuremath{\mathbb{t}}}
\newcommand{\BC}{\ensuremath{\mathcal{B}}}

\newcommand{\DC}{\ensuremath{\mathcal{D}}}

\newcommand{\FC}{\ensuremath{\mathcal{F}}}

\newcommand{\HC}{\ensuremath{\mathcal{H}}}

\newcommand{\LC}{\ensuremath{\mathcal{L}}}

\newcommand{\NC}{\ensuremath{\mathcal{N}}}
\newcommand{\OC}{\ensuremath{\mathcal{O}}}
\newcommand{\PC}{\ensuremath{\mathcal{P}}}
\newcommand{\QC}{\ensuremath{\mathcal{Q}}}
\newcommand{\RC}{\ensuremath{\mathcal{R}}}

\newcommand{\YC}{\ensuremath{\mathcal{Y}}}
\newcommand{\ZC}{\ensuremath{\mathcal{Z}}}
\newcommand{\DF}{\ensuremath{\mathfrak{D}}}
\newcommand{\gf}{\ensuremath{\mathfrak{g}}}
\newcommand{\su}{\ensuremath{\mathfrak{su}}}
\newcommand{\so}{\ensuremath{\mathfrak{so}}}
\newcommand{\spin}{\ensuremath{\mathfrak{sp}}}
\newcommand{\gl}{\ensuremath{\mathfrak{gl}}}

\def\thinline{\noalign{\hrule height.2pt}}
\newcommand{\dell}[2]{
  \ensuremath{\frac{\partial #1}{\partial #2}}
}
\newcommand{\ddelta}[2]{
  \ensuremath{\frac{\delta #1}{\delta #2}}
}
\newcommand{\medoplus}{
  \ensuremath{\scalebox{1.}{$\bigoplus$}}
}
\newcommand{\pair}[2]{
  \ensuremath{\left\langle #1, #2 \right\rangle}
}
\newcommand{\pairing}{
  \ensuremath{\left\langle -, - \right\rangle}
}
\newcommand{\lie}[1]{
  \ensuremath{\LC_{#1}}
}
\newcommand{\ket}[1]{
  \left| {#1} \right\rangle
}
\newcommand{\longhookright}{
  \lhook\joinrel\longrightarrow
}
\newcommand{\hookright}{
  \lhook\joinrel\rightarrow
}
\newcommand{\nlra}{
  \mathrlap{
    \ \ \ \raisebox{.12em}{\tiny\textbf{/}}
  }\longrightarrow
}
\DeclareMathOperator\Hom{Hom}

\DeclareMathOperator\Aut{Aut}

\DeclareMathOperator\Mat{Mat}
\DeclareMathOperator\End{End}
\DeclareMathOperator\im{im}
\DeclareMathOperator\Fun{Fun}
\DeclareMathOperator\Map{Map}
\DeclareMathOperator\U{U}
\DeclareMathOperator\SU{SU}
\DeclareMathOperator\SO{SO}
\DeclareMathOperator\USp{USp}
\DeclareMathOperator\id{id}
\DeclareMathOperator\tr{tr}

\DeclareMathOperator\e{e}
\DeclareMathOperator\effMathOp{eff}
\newcommand{\eff}{{\ensuremath{{\effMathOp}}}}
\DeclareMathOperator\dRMathOp{d}
\newcommand{\dR}{\ensuremath{{\dRMathOp}}}
\DeclareMathOperator\SUGRAMathOp{SUGRA}
\newcommand{\SUGRA}{\ensuremath{{\SUGRAMathOp}}}
\DeclareMathOperator\redMathOp{red}
\newcommand{\red}{\ensuremath{{\redMathOp}}}
\DeclareMathOperator\PAMathOp{I}
\newcommand{\PA}{\ensuremath{{\PAMathOp}}}
\DeclareMathOperator\LIMathOp{L}
\newcommand{\LI}{\ensuremath{{\LIMathOp_\infty}}}
\newcommand{\qcLi}{\ensuremath{{\mathrm{qcL}_{\infty}}}}
\newcommand{\cLIsymp}{\ensuremath{{\mathrm{cL}_{\infty}^{\mathrm{symp}}}}}
\DeclareMathOperator\IEMathOp{IE}
\newcommand{\IE}{\ensuremath{{\IEMathOp}}}
\DeclareMathOperator\ptMathOp{pt}
\newcommand{\pt}{\ensuremath{{\ptMathOp}}}
\DeclareMathOperator\oppMathOp{op}
\newcommand{\opp}{\ensuremath{{\oppMathOp}}}
\DeclareMathOperator\CutMathOp{Cut}
\newcommand{\Cut}{\ensuremath{{\CutMathOp}}}
\DeclareMathOperator\GlueMathOp{Glue}
\newcommand{\Glue}{\ensuremath{{\GlueMathOp}}}
\DeclareMathOperator\GeomMathOp{Geom}
\newcommand{\Geom}{\ensuremath{{\GeomMathOp}}}
\DeclareMathOperator\cnMathOp{N}
\newcommand{\CN}{\ensuremath{{\cnMathOp^*}}}
\DeclareMathOperator\TanMathOp{T}
\newcommand{\Tan}{\ensuremath{{\TanMathOp}}}
\DeclareMathOperator\moduloMathOp{mod}
\newcommand{\Mod}{\ensuremath{{\moduloMathOp}}}
\newtheoremstyle{note}  
{10pt}                  
{10pt}                  
{}                      
{}                      
{\scshape\bfseries}     
{}                      
{1.0em}                 
{}                      
\theoremstyle{note}
\newtheorem{theorem}{Theorem}[section]
\newtheorem{definition}[theorem]{Definition}
\newtheorem{example}[theorem]{Example}
\newtheorem{remark}[theorem]{Remark}
\newtheorem{lemma}[theorem]{Lemma}
\newtheorem{proposition}[theorem]{Proposition}
\newtheorem{conjecture}[theorem]{Conjecture}
\newcommand{\ra}{\rightarrow}
\newcommand{\lra}{\longrightarrow}
\newcommand{\mr}{\mathrm}

\newcommand{\mc}{\mathcal}
\newcommand{\mtt}{\mathtt}
\newcommand{\Ghat}{G}
\newcommand{\Hhat}{H}
\newcommand{\Gcl}{G^\mathrm{cl}}
\newcommand{\Hcl}{H^\mathrm{cl}}

\newcommand{\rel}[1]{\!\!\!\overset{#1}{\nlra}}
\newcommand{\g}{\mathfrak{g}}
\newcommand{\til}{\widetilde}
\newcommand{\Ci}{\mr{C}^\infty}
\newcommand{\Lie}[1]{\mc{L}_{#1}}
\begin{document}
\title{
  \textbf{Towards First Quantisation Formalism for
  AKSZ Theories}
}
\date{}
\maketitle
\thispagestyle{empty}
\begin{abstract}
  \noindent Given an AKSZ theory $\mathds{T}$ on a manifold $M$, with target a graded vector space $Y$,
  we formulate a 1-dimensional theory $\mathbb{t}$ on graphs (the ``first quantisation picture for $\mathds{T}$''), whose partition functions reproduce the Feynman graphs of $\mathds{T}$.
  More precisely, the theory $\mathbb{t}$ is itself a 1d AKSZ theory with the target built out of $M$, and involving a coupling to 1d supergravity.
  It yields a form on the space of metric graphs (with length $T$ of an edge and its de Rham differential $\mathrm{d} T$ interpreted as the zero-modes of the graviton and gravitino, respectively); its integral yields the sum of Feynman graphs of $\mathds{T}$.
  We study the theory $\mathbb{t}$ in the BV-BFV formalism; a gauge-fixing of $\mathds{T}$ corresponds to a gauge-fixing of $\mathbb{t}$.
  At the classical level, $\mathbb{t}$ assigns to vertices certain Lagrangian submanifolds $L_k$ in Cartesian powers $\Phi^{\times k}$ of the phase space $\Phi$ of $\mathbb{t}$.
  These submanifolds can be thought of as defining a cyclic $\mathrm{L}_\infty$-algebra in Weinstein's symplectic category
  (``dequantising'' the cohomological vector field on the target
  of $\mathds{T}$).
  In the path integral construction of $\mathbb{t}$, Lagrangians $L_k$ determine sewing conditions for fields on the incident edges at a $k$-valent vertex.
  We give examples of this paradigm, such as when $\mathbb{t}$ on edges is the Witten-Morse supersymmetric quantum mechanics (which corresponds to a particular type of gauge-fixing for $\mathds{T}$ and $\mathbb{t}$).
  In the example where $\mathds{T}$ is the non-abelian Chern--Simons
  theory with structure Lie algebra $\mathfrak{su}(2)$, we describe the vertex Lagrangian  $L_{\mathrm{W}}$ (the ``Wigner Lagrangian'' ).
\end{abstract}
\newpage
\tableofcontents
\section*{Acknowledgements}
\addcontentsline{toc}{section}{Acknowledgements}
The authors thank Alberto S. Cattaneo and Konstantin Wernli for discussions that inspired this project.
A significant part of the research was conducted while the first author participated in the programme “Cohomological Aspects of Quantum Field Theory” (spring 2025) at the Mittag-Leffler Institute in Djursholm, Sweden, supported by the Swedish Research Council under grant no. 2021-06594.
L.M. thanks Anton Alekseev for inspiring discussions during the programme.

The authors acknowledge partial support of the SNF Grant No. 200021\_227719, of the Simons Collaboration on Global Categorical Symmetries, and of the FIM at ETH Zurich.
This research was (partly) supported by the NCCR SwissMAP, funded by the Swiss National Science Foundation.
\section{Introduction}
\label{sec:Introduction}
In this note we present a $1$-dimensional model for computing the Feynman graphs of
AKSZ theories.\footnote{
  This class of topological sigma-models is due to Alexandrov--Kontsevich--Schwarz--Zaboronsky \cite{Alexandrov:1995kv}.
  Examples include Chern--Simons theory, $BF$ theory, Poisson sigma model.
  For the Poisson sigma model, for the purposes of the present paper, we restrict to the case with target a vector space with a polynomial Poisson bivector.
}
We consider an AKSZ theory $\TB$ with space of fields of the form
\begin{align}
\label{eq:intro_AKSZ_fields}
  \FC^\TB \coloneqq \Map(\Tan[1] M, Y)
  \cong \Omega^\bullet(M;Y).
\end{align}
where $Y$ is some graded vector space\footnote{
    The ground field for $Y$ can be taken to be $\RB$ or $\CB$, see \Cref{rem: ground field R or C}.
}
and $M$ is a closed oriented $m$-manifold.

We assume that $Y$ is equipped with a nondegenerate pairing $\pairing_Y$ of degree $m-1$ and a degree $m$ function $\Theta_Y$ that Poisson commutes with itself.
The classical BV action of $\TB$ is 
\begin{align}
\label{eq:ST_AKSZ_action}
  S^{\TB} = \int_M \frac12 
  \langle A \stackrel{\wedge}{,} \dR A \rangle_{Y}
  + \Theta_Y(A),
\end{align}
where $A \in \Omega^\bullet(M; Y)$ is the field.
Writing $\Theta_Y(A)$ in the integral indicates that the top form-degree part is being picked.
Here, $\Theta_Y$ can be seen as a perturbation of abelian Chern--Simons theory.

The effective action of such a theory can be expressed by its Feynman diagram expansion
\begin{align}\label{intro S^T_eff}
  S^\TB_\eff(a) &= \sum_\Gamma \frac{\hbar^{l(\Gamma)}}{|\Aut(\Gamma)|} F_\Gamma(a),
\end{align}
where $a$ is the residual (or ``infrared'') 
field, the sum is over
connected graphs $\Gamma$, $l(\Gamma)$ is the loop number, $\Aut(\Gamma)$ the automorphism group of $\Gamma$ and $F_\Gamma$ denotes the weight associated to $\Gamma$.

We denote the space of residual fields $\FC^{\TB}_{\mr{eff}}$.
The weight associated to $\Gamma$ in (\ref{intro S^T_eff}) can be written as
\begin{align}
\label{eq:Feynman_weight_definition}
    F_\Gamma(a) \coloneqq \int_{M^{\times \mr{V}(\Gamma)}} \pair{\bigwedge_{e \in \mr{IE}(\Gamma)} \pi^*_e \eta \wedge \bigwedge_{l \in \mr{L}(\Gamma)} \pi^*_{v(l)} \imath(a)}{\bigotimes_{v \in \mr{V}(\Gamma)} c^Y_{\mr{val}(v)}}_\Gamma.
\end{align}
Here:
\begin{itemize}
    \item
        $\mr{V}(\Gamma)$ is the set of vertices of $\Gamma$, $\mr{IE}(\Gamma)$ is the set of internal edges of $\Gamma$, $\mr{L}(\Gamma)$ is the set of leaves of $\Gamma$ (thought of as loose half-edges) with $v(l)$ the incident vertex. 
    \item
        For an edge $e$ the map $\pi_e \colon M^{\times \mr{V}(\Gamma)} \lra M^{\times 2}$ is the projection onto the copies of $M$ associated to the vertices connected by $e$.
        Similarly, for a vertex $v$ the map $\pi_v \colon M^{\times \mr{V}(\Gamma)} \lra M$ is the projection onto the copy of $M$ associated to $v$.
    \item
        $\eta \in \Omega^{m-1}(M \times M; Y \otimes Y)$ is the integral kernel of the propagator $K$ associated to some gauge-fixing of $\TB$.
        For an edge $e$ connecting a vertex $v$ to itself, one 
        replaces $\pi^*_e \eta$ by a suitably regularised diagonal evaluation $\pi^*_v \eta^{\mr{reg}}_\Delta$.
        If $\eta$ is sufficiently nice, one can replace $M^{\times \mr{V}(\Gamma)}$ by $\overline{\mr{C}}_{\mr{V}(\Gamma)}(M)$, the associated Fulton--MacPherson--Axelrod--Singer (FMAS) compactified configuration space.
    \item
        The map $\imath \colon \FC^{\TB}_{\mr{eff}} \xhookrightarrow{} \FC^{\TB}$ denotes the inclusion of residual fields into the space of all fields.
        This map is a part of the strong deformation retraction associated to a good gauge-fixing (cf. \Cref{rem:SDR_from_H} and \Cref{subsubsec:Good_gauge_fixings}).
    \item
        The maps $c_n^Y$ are the coefficients of the Taylor expansion of the degree $m$ perturbation term:
        \begin{align}
            \Theta_Y(y) = \sum_{n} \frac{1}{n!} c^Y_n (\underbrace{y, \ldots, y}_{\text{$n$ times}}),
        \end{align}
        for $y\in Y$.
    \item
        The pairing $\pairing_\Gamma$ denotes the contraction 
        of $Y^{\otimes N}$ and $(Y^*)^{\otimes N}$ prescribed by $\Gamma$, where $N$ is the number of half-edges including leaves.
\end{itemize}
\begin{remark}
\label{rem:Weights_full_vs_residual}
    Instead of the weights of $S^\TB_\mr{eff}(a) \in \Ci(\FC^{\TB}_{\mr{eff}})$, we will construct the weights of the full effective action $S^\TB_\mr{eff, full}(A) \in \Ci(\FC^{\TB})$.
    These weights are identical except for the leaves which – for $S^\TB_\mr{eff, full}$ – are decorated by the projector onto $\imath(\FC^{\TB}_{\mr{eff}}) \subset 
    \FC^\TB
    $ given by $\imath \circ p$ (with $p \colon \FC^{\TB} \twoheadrightarrow \FC^{\TB}_{\mr{eff}}$ the projection onto infrared fields).
    %
    Since this implies $S^\TB_\mr{eff, full}(A) = p^* S^\TB_\mr{eff}$, recovering the weights of one effective action determines those of the other.
\end{remark}

Our goal is to find a $1$-dimensional field theory $\tb$ such that its space of states at a point satisfies $\HC^{\tb}_{\pt} = \FC^{\TB}$ and whose partition function evaluated on some graph $\Gamma$ reproduces the value of $F_\Gamma$.
To define such a field theory, we will need to make precise the following decorations of $\Gamma$:
\begin{enumerate}[label=\Roman*)]
  \item
    \textbf{Assign the propagator $K$ to internal edges.}
    In \Cref{sec:1D_AKSZ} we will discuss how certain gauge-fixings of a $1$-dimensional AKSZ theory $\tb$ on an interval recover the propagators of an AKSZ theory $\TB$.
    %

  \item
    \textbf{Assign residual fields to leaves.}
    These can be recovered from $\tb$ on an interval by letting the length of the interval go to $+ \infty$.
    In this limit, the partition function of $\tb$ tends to the projector onto the kernel of $[\dR, K]$ with $K$ the propagator, which will be exactly the space of residual fields of an appropriately gauge-fixed AKSZ theory $\TB$.

  \item
    \textbf{Assign vertex tensors to vertices.}
    E.g. for non-abelian Chern--Simons theory we should assign the Lie bracket.
    In \Cref{sec:DequantLie} we will sketch a paradigm for geometric data corresponding – via geometric quantisation – to a collection of vertex tensors.
    We will apply this paradigm in \Cref{sec:1D_theory_graphs} to formulate sewing conditions for boundary fields of the $1$-dimensional AKSZ theory $\tb$ on several intervals.
    We also give an explicit example for $\su(2)$-valued Chern--Simons theory and comment on generalisations to other Lie algebras and to theories with higher operations (e.g. the Poisson sigma model).
\end{enumerate}
An example of this procedure is depicted in \Cref{fig:CStree_decomp}.
{\begin{figure}[H]
  \centering
  \includegraphics[width=.9\columnwidth]{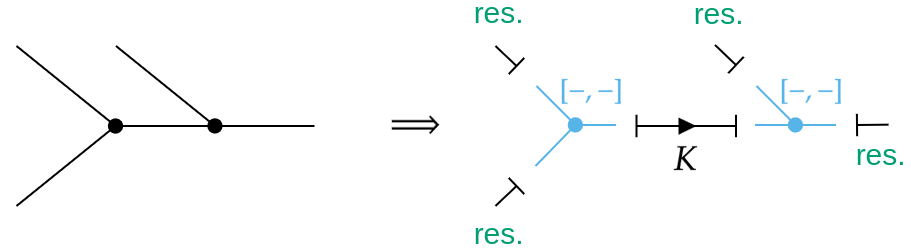}
  \caption{
    Decoration of a 
    Feynman graph appearing in the effective action of non-abelian Chern--Simons 
    theory with coefficients in a quadratic Lie algebra $\g$.
    Internal edges are decorated by a propagator $K$, interaction vertices (vertex tensors) are given by wedging forms and taking the Lie bracket in $\g$.
  }
  \label{fig:CStree_decomp}
\end{figure}}
\noindent The plan of the paper is as follows: \\

\textbf{{\Cref{sec:HTQFT}}:}
    We review Andrey Losev's notion of HTQFT \cite{Losev:2019bel, Beck:2024xtd} which serves as a guiding principle to formulate a $1$-dimensional theory on graphs.\footnote{
      HTQFT stands for \textit{Higher Topological Quantum Field Theory} while the $1$-dimensional case, the HTQM, stands for \textit{Higher Topological Quantum Mechanics}.
    } \\

\textbf{{\Cref{sec:1D_AKSZ}}:}
    To obtain a $1$-dimensional HTQFT on an interval in the BV-BFV formalism \cite{Cattaneo:2012qu, Cattaneo:2015vsa, Cattaneo:2017tef}, we construct and gauge-fix a $1$-dimensional AKSZ theory $\tb$.
    We discuss good gauge-fixing data and show that the \eqref{eq:mQME} (\textit{modified Quantum Master Equation}) for this theory recovers the \eqref{eq:HTQM_closed} equation of HTQM. \\

\textbf{{\Cref{sec:DequantLie}}:}
    In this section we propose an ansatz for a classical 
    (geometric)
    datum corresponding – via geometric quantisation – to vertex tensors.
    Using coadjoint orbits, we formulate such data for the Lie bracket of $\su(2)$.
    We also sketch tentative generalisations to more general Lie algebras and $\LI$-algebras. \\

\textbf{{\Cref{sec:1D_theory_graphs}}:}
    We show how to augment the $1$-dimensional theory $\tb$ on intervals from \Cref{sec:1D_AKSZ} by sewing conditions representing vertex tensors obtained in \Cref{sec:DequantLie}.
    This allows us to define a theory on graphs which reproduces the weights of Feynman graphs of 
    an AKSZ theory $\TB$. \\

Originally, the goal of this work was to find a first quantisation formalism (see \cite[section 3]{Dijkgraaf:1997ip} and \cite[section 1]{Contreras:2023huh}) for Chern--Simons theory.
More concretely, we wanted to understand the sum of its Feynman graphs as a degeneration (tentatively a ``zero-width limit'') of Witten's A-model coupled to $2$-dimensional topological gravity \cite{Witten:1988xj}.

And indeed
one recovers first-quantised Chern--Simons theory, with a class of admissible gauge-fixings.\footnote{
  As alluded to in \cite{Chekeres:2021ieg}, one can use this to recover a loop-enhancement of the Fukaya--Morse $A_\infty$ category by a special axial gauge (see \Cref{subsubsec:CS_FukayaMorseWitten}).
  Making this precise is work in progress.
}
The hope is that this serves to better understand Witten's idea \cite{Witten:1992fb} that Chern--Simons theory is a topological string theory.
In view of our findings, a large class of AKSZ theories should allow for a similar relation.
\section{HTQFT}
\label{sec:HTQFT}
In \Cref{subsec:Cutting_Gluing} and \Cref{subsec:HTQFT} we give a brief review of Andrey Losev's notion of HTQFT, following \cite{Losev:2019bel}.
We will use the $1$-dimensional case, the HTQM, as a running example.
Its generalisation to metric graphs will be presented in \Cref{subsec:HTQFT_graphs}.
For details, we refer the reader to the original text, see also \cite[Section 3]{Beck:2024xtd}.
\subsection{Cutting and Gluing}
\label{subsec:Cutting_Gluing}
An HTQFT \cite{Losev:2019bel} is a monoidal functor
\begin{align}
  \ZC \colon \mathtt{Cob}_{\mathrm{dec}}^\sqcup \longrightarrow \mathtt{Vect}^\otimes,
\end{align}
%
from the category of oriented $n$-dimensional cobordisms, decorated by some geometric structure (the exact type is model-dependent), to the category of vector spaces and linear maps between them.
To an $(n-1)$-dimensional manifold $\Sigma$ the functor $\ZC$ assigns some vector space $\ZC(\Sigma)$, such that $\ZC(\Sigma^\opp) = \ZC(\Sigma)^*$.
For a fixed $n$-manifold $X$ one has
\begin{align}
  \ZC(X) \in \ZC(\partial X) \otimes \Fun(\Geom_X) .
\end{align}
$\Geom_X$ is the space of geometric data of given type (e.g. metrics) on $X$.

\textbf{Cutting along an $(n - 1)$-manifold $\Sigma$} (non self-intersecting, not crossing $\partial X$) is assumed to give a functorial prescription $\Cut_\Sigma \colon \Geom_X \rightarrow \Geom_{X^\Sigma}$, where $X^\Sigma
$ 
is $X$ cut along $\Sigma$.
One has the pullback map
\begin{align}
  \Cut_\Sigma^* \colon \Fun(\Geom_{X^\Sigma}) \longrightarrow \Fun(\Geom_X).
\end{align}
Since $\partial X^\Sigma = \partial X \sqcup \Sigma \sqcup \Sigma^\opp$ we have (by monoidality of $\ZC$) $\ZC(\partial X^\Sigma) = \ZC(\partial X) \otimes \ZC(\Sigma) \otimes \ZC(\Sigma)^*$.
We get a natural map
\begin{align}
  \Glue_{X^\Sigma} \colon \ZC(\partial X^\Sigma) \longrightarrow \ZC(\partial X).
\end{align}

The \textbf{gluing property of HTQFT} is then given by:
\begin{align}
\label{eq:CutGlue}\tag{Gluing}
  (\Glue_{X^\Sigma} \otimes \Cut_\Sigma^*) (\ZC(X^\Sigma)) = \ZC(X).
\end{align}
In lax terms, cutting the manifold and gluing it back together yields the same value of $\ZC$.
\begin{example}[Quantum Mechanics]
\label{ex:QM}
  Take $X = I$ and $\Geom_X = \RB_+$ (a metric/length on $I$).
  Then
  \begin{align}
    \ZC(\partial I) &= \ZC(\pt) \otimes \ZC(\pt^\opp) = V \otimes V^* = \End(V),
  \end{align}
  so that
  \begin{align}
    \ZC(I) \in V \otimes V^* \otimes \Fun(\RB_+) = \Map(\RB_+, \End(V)).
  \end{align}
  Then from \eqref{eq:CutGlue} one can deduce that $\ZC(I)$ is a representation of $\RB_+$ in $\End(V)$.
  Denote by $\ZC(I; T)$ the evaluation of $\ZC(I)$ at $T \in \RB_+$.
  Assuming $\lim_{T \rightarrow 0} \ZC(I; T) = \id_V$ as well as differentiability at $T = 0$, we can conclude
  \begin{align}
    \ZC(I; T_0 + T_1) &= \ZC(I; T_1) \circ \ZC(I; T_0) \qquad \Longleftrightarrow \qquad \ZC(I; T) = \e^{- T H}, \quad H \in \End(V).
  \end{align}
\end{example}
\subsection{Topological Quantum Field Theory}
\label{subsec:HTQFT}

To formulate what it means for a theory to be \textit{topological}, we replace $V$ by a cochain complex $(V^\bullet, Q)$ and $\Geom$ by $\Tan[1]\Geom$.
For a fixed $n$-dimensional manifold $X$ the functor $\ZC$ then takes values in
\begin{align}
  \ZC(X) \in \ZC(\partial X) \otimes \Omega^\bullet(\Geom_X) .
\end{align}
\begin{definition}[Topological Theories]
  We call a theory defined by a functor $\ZC$ as described above \textbf{topological} if it satisfies the \textbf{closedness axiom}
  \begin{align}
  \label{eq:HTQM_closed}\tag{Closedness}
    (\dR_{\Geom} + Q) \ZC(X) = 0.
  \end{align}
\end{definition}

We can expand \eqref{eq:HTQM_closed} in form-degree along $\Geom$.
The first two equations are
\begin{align}
  Q \ZC^{(0)}(X) &=0, \\
  \dR_{\Geom} \ZC^{(0)}(X) &= - Q \ZC^{(1)}(X).
\end{align}
I.e. the degree $0$ part of $\ZC$ is $Q$-closed and its change w.r.t. the geometric data is $Q$-exact.
Thus, the class of $\ZC^{(0)}$ in cohomology of $Q$ is (locally) constant on $\mr{Geom}$.
This recovers the usual notion of topological quantum field theory \`a la Atiyah. 
\begin{example}[Higher Topological Quantum Mechanics]
  Under the same assumptions as in \Cref{ex:QM}, one can show that the most general HTQM on $X = I$ with $\Geom = \RB_+$ is given by
  \begin{align}
  \label{eq:FormUniversalTQFT_1D}
    \ZC(I; T, \dR T) &\coloneqq \e^{[\dR_I + Q, - T G]} = \e^{- T H - \dR T G}\qquad \in \mr{End}(V)\otimes \Omega^\bullet(\RB_+),
  \end{align}
  where $\dR_I=\dR T\frac{\dR}{\dR T}$ is the de Rham operator on $\RB_+$, $G$ is an additional differential (of degree $-1$) on $V$ and $H \coloneqq [Q, G]$.

  Let $\ZC(\pt) = (\Omega^\bullet(M), Q = \dR_M)$ for some manifold $M$.
  There are two special cases of importance for us, and an interpolation between them:
  \begin{enumerate}[label=\Roman*)]
    \item \textbf{Hodge Laplacian:} Choose $G = \dR^*$, so that
    \begin{align}
      H = [\dR , \dR^*] = \Delta,
    \end{align}
    the Laplace operator.

    \item \textbf{Pure Morse TQM:} Let $G = \imath_v$, where $v$ is the gradient vector field of some Morse function on $M$.\footnote{
      These are analytically well-behaved, however in principle one can pick more general vector fields.
      Interesting examples are Anosov vector fields \cite{Schiavina:2023ysh}.
    }
    Then
    \begin{align}
      H = [\dR, \imath_v] = \lie{v},
    \end{align}
    the Lie derivative along $v$.
    %
    \item \textbf{Witten's Morse TQM (interpolation).}
      Fix $\varepsilon>0$. Set
      \begin{equation}
      \label{eq:G_MorseWitten_Ex2.3}
         G = \imath_v + \varepsilon \dR^*.
      \end{equation}
      Then
      \begin{equation}
      \label{eq:H_MorseWitten_Ex2.3}
         H = \lie{v} + \varepsilon \Delta.
      \end{equation}
    This is the data of Witten's supersymmetric quantum mechanics \cite{Witten:1982im}, see also \cite{Frenkel:2006fy}.
    Limits $\varepsilon\ra \infty$ and $\varepsilon\ra 0$ correspond to the two cases above.
  \end{enumerate}
\end{example}
\subsection{Topological Quantum Mechanics on Metric Graphs}
\label{subsec:HTQFT_graphs}
%
Let $V = \ZC(\pt)$ be a cochain complex with differential $Q$, equipped with the second differential $G$ of degree $-1$, and a nondegenerate pairing $\pairing$ of degree $-1$, such that $G, Q$ are 
self-adjoint (i.e. $\pair{Q(-)}{-} = \pair{-}{Q(-)}$ and $\pair{G(-)}{-} =  \pair{-}{G(-)}$) with respect to $\pairing$ (and so the Hamiltonian $H = [Q, G]$ and the TQM partition function \eqref{eq:FormUniversalTQFT_1D} are self-adjoint as well).\footnote{
  The pairing $\pairing$ allows us to revert the orientation of the interval on which the TQM is defined, or equivalently view it as defined on unoriented intervals.
  In terms of the second-quantised theory $\TB$, the pairing $\pairing$ is the degree $-1$ symplectic form $\omega^\TB$ on the space of fields and thus corresponds to the Poincaré pairing 
  of forms on $M$ tensored with the target pairing $\pairing_Y$.
}
Fix the structure maps (``cyclic operations'') 
%
\begin{align}
  c_n^k \colon S^n V\ra \CB.
\end{align}
Here the lower index is the arity $n\geq 0$ of the multilinear operation, the upper index $k\geq 0$ is used to denote the loop number.
All maps are of degree zero.
%
%
\begin{definition}[{$\qcLi$-algebra}]
\label{def:L_algebra}
  If the BV action
  \begin{align*}
    S(A) \coloneqq
      \frac12 \pair{A}{Q(A)} + \sum_{n, k} \frac{\hbar^k}{n!} c_n^k(\underbrace{A, \ldots, A}_n)
      , \quad A \in V, 
  \end{align*}
  associated to the maps $\{ c_{\bullet}^{\bullet} \}$ satisfies the Quantum Master Equation (QME)
  \begin{equation}
  \label{eq:def_of_L_algebra_QME}
    \frac{1}{2} \{S, S\} + \hbar \Delta S = 0,
  \end{equation}
  then we call the tuple $(V,Q, \{ c_{\bullet}^{\bullet} \},\pairing)$ a \textbf{quantum cyclic $\LI$-algebra on $V$} (abbreviated as \textbf{$\qcLi$-algebra}).
  Here $\{-, -\}$, $\Delta$ are the degree $+1$ Poisson bracket and the BV Laplacian on function on $V$ associated to the symplectic form on $V$ induced by the pairing $\pairing$.\footnote{
    For more extensive discussions and constructions of these algebraic structures, see \cite{Mnev:2008sa, granåker2008unimodular}, \cite[Theorem 8.1]{Cattaneo:2017tef} for the $1$-loop case and \cite{Doubek:2017naz} for the higher-loop case.
  }
  \footnote{
    A \emph{classical} cyclic $\LI$-algebra (as opposed to quantum one) consists of a complex $V,Q$ with a pairing $\pairing$ of degree $N-1$ (for some degree $N \in \NB$) and a collection of maps $c_n \colon S^n V\ra \CB$ of degree $N$, such that $S(A) \coloneqq \frac12 \langle A,Q A \rangle+\sum_n \frac{1}{n!} c_n(A,\ldots,A)$ satisfies the classical master equation $\{S,S\}=0$.
  }
\end{definition}

For simplicity, we will be assuming that $c^0_n = 0$ for $n = 0, 1, 2$.
One can consider ``oriented operations'' obtained by converting a subset of inputs in a cyclic operation $c_n^k$ into outputs using the pairing:
\begin{align}
  l_{p,q}^k \colon S^p V \xrightarrow{c_{p+q}^k} S^q V^* \xrightarrow{(\pairing^{-1})^{\otimes q}} S^q(V[1]).
\end{align}
The $l_{p,q}^k$ are multilinear operations with $p$ inputs and $q$ outputs (note that $l_{n,0}^k = c_n^k$).
Then ${(V,\{l^0_{n, 1} + Q \delta_{n,1}\}_{n \geq 0})}$ is a (generally, curved) $\LI$-algebra, which is flat if $c_1^0 = 0$.\footnote{\label{footnote: grading convention for L-infty algebras}
    Here we are using the grading convention for $\LI$-algebras where operations are graded-symmetric and have degree $1$.
    This means that on $V[-1]$ one has operations which are graded-skew-symmetric and have degree $2 - \mr{arity}$ (which is the more common convention for $\LI$-algebras).
    Thus, e.g., for $\g$ an ordinary Lie algebra (concentrated in degree zero), $V = \g[1]$ is an $\LI$-algebra in our convention.
}
%
\begin{example}[\texorpdfstring{$\LI$}{L∞}-relations]
\label{ex:Hty_relations}
  We give here examples for the $k = 0$ and $k = 1$ relations of a $\qcLi$-algebra.
  We assume for simplicity that all operations have at most $1$ output.
  In the following we denote by $v_i$ the inputs and by $\pm$ the appropriate Koszul sign. \\
  
  \noindent \textbf{Homotopy Jacobi identities ($\hbar^0$):} For each $n \geq 0$ we have
  \begin{align}
  \label{eq:Hty_Jacobi}
    \sum_{\sigma \in \mathrm{S}_n} \sum_{r + s = n} \frac{1}{r! s!} \pm l^0_{r+1, 1}(v_{\sigma_1}, \ldots, v_{\sigma_r}, l^0_{s, 1}(v_{\sigma_{r + 1}}, \ldots, v_{\sigma_n})) = 0.
  \end{align}
  \textbf{Homotopy unimodularity identities ($\hbar^1$):} For each $n \geq 0$ we have
  \begin{align}
  \label{eq:Hty_unimodularity}
    \sum_{\sigma \in \mathrm{S}_n}
    \left(
      \pm \frac{1}{n!} \mathrm{Str}_V l^0_{n + 1, 1}(
        v_{\sigma_1}, \ldots, v_{\sigma_n}, -
      )
      \pm \sum_{r + s = n} \frac{1}{r! s!} l^1_{r + 1, 0}(
        v_{\sigma_1}, \ldots, v_{\sigma_r}, l^0_{s, 1}(
          v_{\sigma_{r + 1}}, \ldots, v_{\sigma_n}
        )
      )
    \right) = 0.
  \end{align}
\end{example}
From now on – if not explicitly stated otherwise – ``graph'' (or ``metric graph'') refers to a weighted, metric graph:
%
Each internal edge is decorated by geometric data $T_i, \dR T_i$ and each vertex by a non-negative integer weight – the loop number.

Given a $\qcLi$-algebra $(V,Q, \{ c_{\bullet}^{\bullet} \},\langle-,-\rangle)$, we extend the functor $\ZC$ from intervals to graphs as follows:
\begin{itemize}
    \item
        Choose (arbitrarily, cf. \Cref{rem:Auxiliary_orientation}) an orientation on the graph.
    \item
        To \textbf{edges}, assign the value of $\ZC(I; T_i, \dR T_i)$.
    \item
        To each \textbf{vertex} with $p$ incoming edges, $q$ outgoing edges and weight $k$, assign the respective oriented operation $l_{p,q}^k \in \Hom(S^{p} V, S^q (V[1]))$ of the $\qcLi$-algebra.
    \item
        \textbf{Contract} inputs and outputs of the maps according to the composition laws prescribed by the oriented graph.
        An example is given in \Cref{fig:Y_shaped_graph}.
\end{itemize}
{\begin{figure}[H]
  \centering
  \includegraphics[width=1.0\columnwidth]{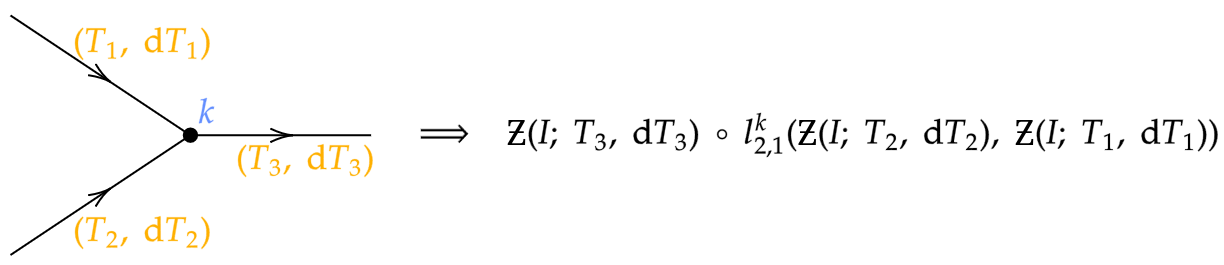}
  \caption{Example of a \textbf{Y}-shaped graph with the only vertex being of weight $k$.}
  \label{fig:Y_shaped_graph}
\end{figure}}

Note that for an oriented graph $\Gamma$ with $N_\mr{in}$ incoming exterior edges  and $N_\mr{out}$ outgoing exterior edges and $\mr{IE}(\Gamma)$ internal edges, the contraction above is an element in
\begin{align}
\label{eq:Hom_space_partitionfunction_graph}
  \mr{Hom}(V^{\otimes N_\mr{in}},(V[1])^{\otimes {N_\mr{out}}})\otimes \Omega^\bullet(\RB_+^{\mr{IE}(\Gamma)}).
\end{align}

Changing the orientation of internal edges of the graph leaves the contraction invariant, while changing the orientation of an external edge reassigns an in-$V$ factor in \eqref{eq:Hom_space_partitionfunction_graph} above as an out-$V$ factor or vice versa using the pairing $\pairing$.
\begin{remark}
[Forgetting the orientation of edges / orienting half-edges toward the incident vertex]
\label{rem:Auxiliary_orientation}
    An equivalent description of HTQM on graphs is as follows.
    Do not orient the graph $\Gamma$, assign $\ZC\circ \pairing^{-1}\in (V\otimes V)\otimes \Omega^\bullet(\RB_+)$ to each internal edge,\footnote{$\ZC\circ \pairing^{-1}$ is an element of total (de Rham plus internal) degree one.}
    assign $\ZC$ to each external edge and assign $c_n^k$ to each $n$-valent vertex of loop weight $k$, and then contract according to graph combinatorics.
    In this picture, all external edges are implicitly treated as incoming.
\end{remark}

If $H \coloneqq [Q, G]$ is positive semi-definite, one can take the length of edges to $+ \infty$.
The value of $\ZC(I; T, \dR T)$ then tends to the projector $\Pi \coloneqq \imath \circ \pi$ onto $\ker(H) \subset V$ as $T \rightarrow \infty$.
Here we denote by $\imath \colon \ker(H) \hookright V$ the canonical inclusion and by $\pi \colon V \twoheadrightarrow \ker(H)$ the projection along $\im(H)$.
\begin{remark}[SDR from \texorpdfstring{$H$}{H}]
\label{rem:SDR_from_H}
    Assume additionally that $G$ vanishes on $\ker(H)$.
    Then the differential $Q$ on $V$ induces a differential $Q' \coloneqq \pi \circ Q \circ \imath$ on $\ker(H)$ and the triple $(\imath, \pi, K \coloneqq \int_{\RB_+} \ZC(I; T, \dR T))$ is a strong deformation retraction (SDR) of $V$ onto $\ker(H)$, i.e.
    \begin{itemize}
        \item
            $\imath,\pi$ are chain maps.
        \item
            $K$ is a chain homotopy between identity and projection onto $\ker(H)$, i.e. $Q \circ K + K \circ Q = \id - \imath \circ \pi$.
        \item
            The side conditions are satisfied: $K^2 = 0$, $\pi \circ K = 0$, $K \circ \imath = 0$.
    \end{itemize}
\end{remark}
\begin{definition}[Pre-Amplitudes]
\label{def:Preamplitudes}
  We define the \textbf{pre-amplitude ${\PA^\infty}$ of the HTQM} by evaluating $\ZC$ on a metric graph and letting all external edges tend to length $+ \infty$.
  We will call infinitely long external edges ``leaves''.
  Thus the pre-amplitude assigned to a graph $\Gamma$ with $\IE(\Gamma)$ internal edges, $p$ incoming and $q$ outgoing leaves is a differential form
  \begin{align} \label{eq:Iinfty_of_Gamma}
    {\PA^\infty}(\Gamma) \in \Omega^\bullet(\RB^{\IE(\Gamma)}_+) \otimes \Hom(\ker(H)^{\otimes p}, (\ker(H)[1])^{\otimes q})
  \end{align}
  of total degree $\IE(\Gamma)$.
\end{definition}

Consider the moduli space $\mathcal{MG}_{p,q}^k$ of connected metric weighted graphs with $p$ unordered incoming leaves, $q$ unordered outgoing leaves and $k$ the total loop number (the number of loops in a graph plus the sum of weights of vertices); internal edges are unoriented. $\mathcal{MG}_{p,q}^k$ is an orbi-cell complex
\begin{equation}
    \mc{MG}_{p,q}^k = \bigsqcup_{\Gamma \in \mr{Graph}_{p,q}^k} \left( \RB_{\geq 0}^{\mr{IE}(\Gamma)} /\!/ \Aut(\Gamma) \right) \; \Big/ \; \sim.
\end{equation}
Here $\sim$ is the equivalence relation identifying a graph $\Gamma$ with an edge $e$ of length zero connecting vertices $u$ and $v$ with either
\begin{enumerate}[label=(\alph*)]
    \item
      the graph $\Gamma$ with edge $e$ collapsed (if $u \neq v$ as vertices in $\Gamma$) or
    \item
      with the graph $\Gamma$ with edge $e$ removed and the loop weight of $v$ increased by one (if $u = v$).
\end{enumerate}

An orbi-cell  $\RB_{\geq 0}^{\mr{IE}(\Gamma)} /\!/ \Aut(\Gamma)$ is the action groupoid for the graph automorphism group (allowing to permute the in-leaves among themselves and out-leaves among themselves) acting on the space of lengths of internal edges $\RB_{\geq 0}^{\mr{IE}(\Gamma)}$.

\begin{definition}[Pre-Amplitudes: Invariant Definition]
\label{rem:Preamplitudes_InvarDef}
  We define the \textbf{pre-amplitude} as a distributional form on the moduli space of metric graphs
  \begin{align}
  \label{eq:Iinfty_distrform_MG}
    ({\PA^\infty})_{p, q}^k \in \Omega^\bullet_\mathrm{distr} \left( \mathcal{MG}_{p, q}^k, \Hom(S^p \ker(H), S^q (\ker(H)[1])) \right),
  \end{align}
  invariant under the graph automorphism groups (acting on $\mc{MG}_{p,q}^k$ and also by permuting the inputs and outputs).
  It is constructed as
  \begin{equation}
  \label{eq:Iinfty_sumof_pushfwds}
      ({\PA^\infty})_{p, q}^k \coloneqq \sum_{\Gamma\in \mr{Graph}_{p,q}^k}(i_\Gamma)_* \PA^\infty(\Gamma),
  \end{equation}
  the sum of pushforwards of forms \eqref{eq:Iinfty_of_Gamma} along the inclusions of open orbi-cells
  \begin{align*}
    \imath_\Gamma\colon \RB_+^{\mr{IE}(\Gamma)}/\!/\mr{Aut}(\Gamma) \longhookright \mc{MG}_{p,q}^k.
  \end{align*}
\end{definition}
\begin{lemma}
    \label{lem:TotalClosedness_Iinfty}
    The distributional form \eqref{eq:Iinfty_distrform_MG}, \eqref{eq:Iinfty_sumof_pushfwds} satisfies
    \begin{align}
      (\dR_\Geom + Q') \ {\PA^\infty} &= 0.
    \end{align}
\end{lemma}
\begin{proof}
    First, consider the case when only the operation $c_3^0$ on $V$ is nonvanishing and all other operations $c_n^k$ vanish.\footnote{
        I.e. $V$ is a differential graded Lie algebra with invariant pairing – the case pertinent to $\TB$ being Chern--Simons or $BF$ theory.
    }
    Then:
    \begin{enumerate}[label=(\roman*)]
        \item
            $c_3^0$ is $Q$-closed.
        \item
            The form $\PA^\infty$ is supported on 3-valent graphs with leaves allowed and is non-distributional (does not have distributional contributions supported at higher-codimension cells of $\mc{MG}$).
        \item \label{lem:TotalClosedness_Iinfty_point3}
            $(\dR_\mr{Geom} + Q') \PA^\infty = 0$ on top cells of $\mc{MG}$ (corresponding to 3-valent graphs), since the operator $(\dR_\mr{Geom} + Q)$ annihilates the edge factors $\ZC(I; t_i, \dR t_i)$ in $\PA^\infty$ and $Q$ annihilates (in the sense ``commutes with'') the leaf factors $\imath$ and annihilates the vertex factors $c_3^0$.
        \item
            Acting with $\dR_\mr{Geom}$ on $\PA^\infty$ does not generate a $\delta$-density supported at higher-codimension cells (edge collapses), due to Jacobi identity on $c_3^0$.
    \end{enumerate}

    A depiction of the relevant cells of the moduli space for this case is given in \Cref{fig:ModuliSpace_cell}.
    {\begin{figure}[H]
      \centering
      \includegraphics[width=.7\columnwidth]{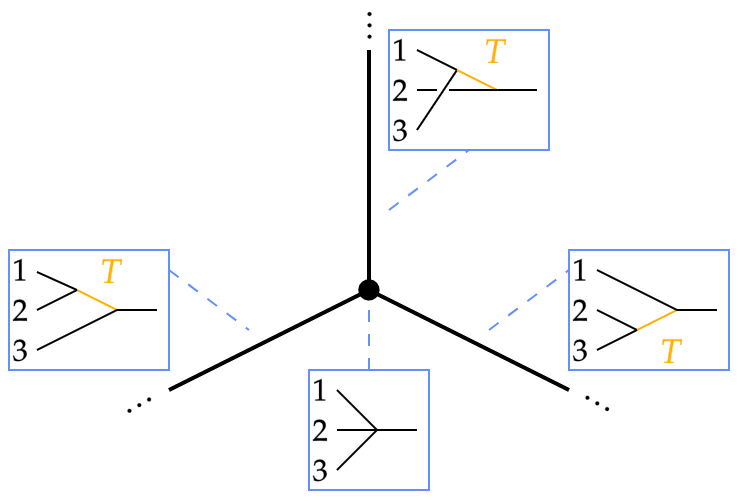}
      \caption{
        Decoration of three $1$-cells by graphs with one internal edge and ordered inputs, glued along a $0$-cell decorated by a graph with no internal edges.
        All vertices are assumed to be of weight $0$ and we introduce an \textit{arbitrary} orientation from left to right.
        The graph decorating the $0$-cell corresponds to $c^0_4$ and vanishes in the above example, making $I^\infty$ non-distributional.
        Collapsing the internal edges of the graphs labelling the $1$-cells yields the Jacobi identity on $c^0_3$.
      }
      \label{fig:ModuliSpace_cell}
    \end{figure}}
    \noindent This proves the special case.
    Now consider the case of $V$ a general quantum cyclic $\LI$-algebra.
    Three points in the argument above change:
    \begin{enumerate}[label=(\arabic*)]
      \item
        $\PA^\infty$ gets distributional contributions on higher-codimension strata,
      \item
        vertex factors $c_n^k$ in $\PA^\infty$ are no longer $Q$-closed,
      \item
        the Jacobi identity is replaced by the more complicated structure equations of a $\qcLi$-algebra.
    \end{enumerate}

    Let $\phi \in \Omega^\bullet(\mc{MG}_{p,q}^k, \mr{Hom}(S^p V, S^q V[1])^*)$ be a ``test form'', which is compactly supported, invariant under graph automorphisms, smooth on orbi-cells, and has well-defined smooth pullbacks by inclusion of higher-codimension strata.
    We have
    \begin{multline}\label{eq:Testform_computation_closedness}
         \int_{\mc{MG}_{p,q}^k}\langle \phi,(d_\mr{Geom}+Q') \PA^\infty\rangle =
         \int_{\mc{MG}_{p,q}^k}-\langle (d_\mr{Geom}+Q')\phi, \PA^\infty\rangle\\
         =
         \sum_{\Gamma\in \mr{Graph}_{p,q}^k}
         \int_{\RB_+^{\mr{IE}(\Gamma)}/\!/ \mr{Aut}(\Gamma)}
         -\langle (d_\mr{Geom}+Q')\phi, \PA^\infty\rangle\\
         =
         \sum_{\Gamma\in \mr{Graph}_{p,q}^k} \frac{1}{|\mr{Aut}(\Gamma)|} \Big(
         \underbrace{\int_{\RB_+^{\mr{IE}(\Gamma)}}
         \langle \phi, (d_\mr{Geom}+Q')\PA^\infty\rangle}_{(a)}-
         \underbrace{\int_{
         \partial_\mr{UV}\RB_+^{\mr{IE}(\Gamma)}}
         \langle \phi, \PA^\infty\rangle}_{(b)} \Big).
    \end{multline}
    Here the first equality is by definition of the derivative of a distribution, for the second we use the orbi-cell decomposition of $\mc{MG}$, and the third holds by Stokes' theorem on cells.
    
    We denote by $\partial_\mr{UV} \RB_+^{\mr{IE}(\Gamma)}$ the union of ``ultraviolet boundary strata'' of the cell, given by setting the length of one of the interior edges to zero.
    Altogether we get the following terms in \eqref{eq:Testform_computation_closedness}:
    \begin{itemize}
      \item[$(a)$]
        For $(d_\mr{Geom} + Q') \PA^\infty$ we sum over vertices $v$ of a graph $\Gamma$ and decorate them with $Q c^{\mr{weight}(v)}_{\mr{val}(v)}$.
        The rest of the internal edges and vertices are decorated by $\ZC$'s and $c_n^k$'s respectively.
        Then one contracts according to graph combinatorics.
        This procedure is a generalization of point \ref{lem:TotalClosedness_Iinfty_point3} above.
      \item[$(b)$]
        The ultraviolet boundary strata yield a sum over graphs $\Gamma'$ with one marked edge $e$ where in the formula for $\PA^\infty$, the $\ZC$-factor for $e$ is replaced with the identity operator.
    \end{itemize}
    
    Then, a contribution $(a)$ for a graph $\Gamma$ and a marked vertex $v$ of $\Gamma$ decorated by $Q c^{\mr{weight}(v)}_{\mr{val}(v)}$ cancels out with the sum of contributions $(b)$ for graphs $\Gamma'$ which differ from $\Gamma$ only in $v$ which is either replaced by two vertices $v', v''$ with weights $k' + k'' = k$ and connected through $e$, or by a single vertex $v'$ with weight $k-1$ and connected to itself via $e$.
    An example of such replacements of $v$ in $\Gamma$ is given in \Cref{fig:GraphSplitting}.
    {\begin{figure}[H]
      \centering
      \includegraphics[width=1.0\columnwidth]{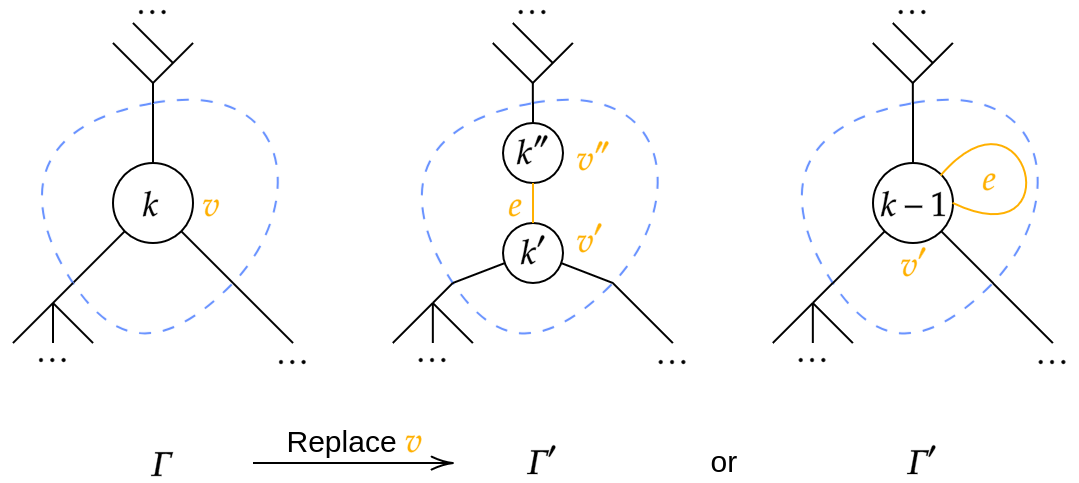}
      \caption{
        Example of replacing a vertex $v$ of weight $k$ in a graph $\Gamma$ by either (middle graph) two vertices connected through a marked edge $e$ and with same total loop number $k'+k''=k$ or by (right graph) a single vertex with loop number decreased by one and connected via $e$ to itself.
      }
      \label{fig:GraphSplitting}
    \end{figure}}

    The cancellation occurs by virtue of structure relations \eqref{eq:def_of_L_algebra_QME} on $c_n^k$:
    \begin{equation}
        \underbrace{\frac{1}{n!} Q c_n^k}_{\text{from $(a)$}} + \underbrace{\sum_{p + q = n + 2} \; \sum_{k' + k'' = k}\frac{1}{p! q!} \{ c_p^{k'}, c_q^{k''} \} + \frac{1}{(n+2)!} \Delta c^{k-1}_{n+2}}_{\text{from $(b)$}} = 0.
    \end{equation}
    Thus, the r.h.s. of \eqref{eq:Testform_computation_closedness} is zero, proving that $(\dR_\Geom + Q') \ {\PA^\infty}$ vanishes as a distributional form.
\end{proof}

\begin{definition}[Amplitudes]
\label{def:Amplitudes}
  For given $p, q$ and $k$, define the \textbf{$(p, q, k)$-amplitude} as
  \begin{align}
    \frac{1}{p! q!} \mathsf{A}_{p, q}^k(\ZC, \{ c_{\bullet}^\bullet \}) \coloneqq
    %
    %
    \sum_{\Gamma \in \mathtt{Graph}_{p, q}^k} \frac{1}{|\Aut(\Gamma)|} \int_{\RB^{\IE(\Gamma)}_+} {\PA^\infty}(\Gamma).
  \end{align}
  Here the sum is over all connected weighted metric graphs with $p$ incoming leaves, $q$ outgoing leaves and $k$ total loop number.
  $\mathsf{A}_{p, q}^k(\ZC, \{ c_{\bullet}^\bullet \})$ is an element of $\Hom(S^p\ker(H), S^q(\ker(H)[1]))$ and of degree $0$.
  
  Following \Cref{rem:Preamplitudes_InvarDef} we may also define it as
  \begin{align}\label{amplitude as integral over MG}
    \frac{1}{p! q!} \mathsf{A}_{p, q}^k(\ZC, \{ c_{\bullet}^\bullet \}) \coloneqq
    %
    %
    \int_{\mathcal{MG}_{p, q}^k} ({\PA^\infty})_{p, q}^k.
  \end{align}
  When the partition function $\ZC$ on intervals and the $\qcLi$-algebra operations $\{ c_{\bullet}^\bullet \}$ are understood, we simply denote the $(p, q, k)$-amplitude by $\mathsf{A}_{p, q}^k$.
\end{definition}

\begin{remark}
    The assignment of leaves of graphs as in- or out- throughout this section is entirely superficial as one can switch between them using the pairing $\pairing$ (cf. \Cref{rem:Auxiliary_orientation}).
    Likewise, the data of ``oriented'' amplitudes $\mathsf{A}^k_{p, q}$ is contained in amplitudes of the form $\mathsf{A}^k_{n, 0}$ with $p + q = n$ – ``cyclic amplitudes''.
    To simplify notation, we will denote $\mathsf{A}^k_{n, 0}$ by $\mathsf{A}^k_n$, $\mc{MG}^k_{n, 0}$ by $\mc{MG}_n^k$ and $\mr{Graph}^k_{n, 0}$ by $\mr{Graph}^k_n$.
\end{remark}
\begin{theorem}
\label{theo:HomotopyTransfer_MG}
  If the maps $\{ c_{\bullet}^{\bullet} \}$ satisfy the relations of a $\qcLi$-algebra on $V$, then the operations $\{ \mathsf{A}_{\bullet}^{\bullet} \}$ satisfy the relations of a $\qcLi$-algebra on $\ker(H)$ (with differential $Q' \coloneqq \pi \circ Q \circ \imath$ and $\pairing'$ the restriction of the pairing $\pairing$ on $V$).
\end{theorem}
\begin{remark}
  Following \Cref{theo:HomotopyTransfer_MG} there are three main perspectives on amplitudes:
  \begin{enumerate}[label=\Roman*)]
    \item
      They describe the homotopy transfer of a $\qcLi$-algebra on $V$ onto the retract $\ker(H)$.
    \item
      We can see them as an instance of topological quantum mechanics coupled to $1$-dimensional topological gravity (the metric data):
      From this point of view, the amplitudes $\mathsf{A}_{p, q}^k$ are computing the scattering amplitudes for a $1$-dimensional theory.
      The $\qcLi$-algebra relations on $\ker(H)$ are then a consequence of those on $V$ together with the factorisation of the pre-amplitude on the ``infrared''
      boundary strata of the moduli space $\mathcal{MG}$ corresponding to making one internal edge infinitely long.
    \item
      The amplitudes can be seen as sums of Feynman graphs of the effective action of a second quantised theory $\TB$ defined by the BV action $S^\TB(A)$ of \Cref{def:L_algebra}.
      The fact that the generating function for the effective theory, $S_\eff^\TB(a)$ satisfies the Quantum Master Equation (as a consequence of BV-Stokes' theorem for fibre BV integrals) implies the $\qcLi$-algebra relations of the operations $\{ \mathsf{A}^\bullet_\bullet \}$ on $\ker(H)$.
  \end{enumerate}
\end{remark}
\begin{proof}[Sketch of the proof of {\Cref{theo:HomotopyTransfer_MG}}]~\\
  %
  %
  %
  We adopt the viewpoint II) above.
  Using 
  \Cref{lem:TotalClosedness_Iinfty}, we have
  \begin{align}
    \label{thm 2.13 computation line 1}
    \frac{1}{n!} Q' \mathsf{A}^{k}_{n} =&~
    \int_{\mc{MG}_{n}^k}((\dR_\mr{Geom}+Q')-\dR_\mr{Geom})\PA^\infty
    \\
    \label{thm 2.13 computation line 2}
    \underset{
    \Cref{lem:TotalClosedness_Iinfty}}{=}&~\int_{\mc{MG}_{n}^k} -\dR_\mr{Geom} \PA^\infty  \\
     \label{thm 2.13 computation line 3}
    =& ~\sum_{\Gamma\in \mr{Graph}_n^k} \frac{1}{|\mr{Aut}(\Gamma)|}\Big(
    -\int_{\RB_+^{\mr{IE}(\Gamma)}} \dR_\mr{Geom}\PA^\infty + \int_{\partial_\mr{UV}\RB_+^{\mr{IE}(\Gamma)}}\PA^\infty
    \Big) \\
     \label{thm 2.13 computation line 4}
    =&~ \sum_{\Gamma\in \mr{Graph}_n^k } -\frac{1}{|\mr{Aut}(\Gamma)|}  \int_{\partial_\mr{IR}\RB_+^{\mr{IE}(\Gamma)}}\PA^\infty \\
     \label{thm 2.13 computation line 5}
    =&~ -\sum_{p+q=n+2}\; \sum_{k'+k''=k} \frac{1}{p!q!} \{\mathsf{A}^{k'}_p, \mathsf{A}^{k''}_q\}'-\frac{1}{(n+2)!}\Delta' \mathsf{A}^{k-1}_{n+2},
  \end{align}
  which proves the structure equation \eqref{eq:def_of_L_algebra_QME} for operations $\{\mathsf{A}_\bullet^\bullet\}$.

  We provide explanations for the computation above: The second term in the third line (the integral over the ultraviolet boundary) appears by the mechanism of \eqref{eq:Testform_computation_closedness}.
  Going to the fourth line, we evaluate the first term in \eqref{thm 2.13 computation line 3} using Stokes' theorem
  – it yields the sum of contributions of the ultraviolet boundary (length of one of the internal edges is $T_e \ra 0$) – which cancels out the second term in \eqref{thm 2.13 computation line 3} – and the infrared boundary (length of one of the internal edges is $T_e \ra \infty$).
  The infrared boundary of a cell of $\mc{MG}$ consists of
  \begin{enumerate}[label=(\alph*)]
      \item
        strata where an edge $e$ (with length $T_e \ra \infty$) is separating $\Gamma$ into two disjoint subgraphs $\Gamma', \Gamma''$ each with a marked leaf $l', l''$.
        On such a stratum the pre-amplitude $\PA^\infty(\Gamma)$ factorises as $\PA^\infty(\Gamma') \circ_{l' l''} \PA^\infty(\Gamma'')$ where $\circ_{l' l''}$ means that the two marked leaves are joined by $\imath \circ \pi$.
      \item
        strata where the infinite edge $e$ is non-separating and cutting it yields a connected graph $\Gamma'$ with two marked leaves $l_1, l_2$.
        On such a stratum, the pre-amplitude $\PA^\infty(\Gamma)$ becomes $\mr{Str}_{l_1 l_2} (\PA^\infty(\Gamma'))$ – the partial supertrace of $\PA^\infty(\Gamma')$ over $\ker(H)$, where $l_2$ seen as an out-leaf is plugged into $l_1$.
  \end{enumerate}
  Integrating $\PA^\infty$ over the lengths of non-infinite edges and summing over graphs, we obtain (\ref{thm 2.13 computation line 5}).
\end{proof}

We illustrate the argument above by two examples of $\qcLi$-algebra relations – one for $0$-loop operations, one involving a $1$-loop operation. We have
\begin{align}
  \frac{1}{3!} Q' \mathsf{A}^{0}_{3, 1} =&~ \frac{1}{3!}\pi \circ (Q l_{3, 1}^0) \circ \imath^{\otimes 3} - \int_{\partial (\mathcal{MG}_{3, 1}^{0,\mr{top}})} ({\PA^\infty})_{3, 1}^{0}.
\end{align}
Here on the r.h.s. the first term is the separated contribution of the zero-dimensional cell of $\mc{MG}_{3,1}^0$ (\Cref{fig:ModuliSpace_3_1_0}, left).
The second term is the integral over the top (1-dimensional) cells of the space of metric graphs $\mc{MG}_{3,1}^0$ evaluated by Stokes' theorem.
The integral on the r.h.s. has both ultraviolet and infrared boundary contributions (cf. \Cref{fig:ModuliSpace_3_1_0}, right):
Together with the first term involving $Q l_{3, 1}^0$ the ultraviolet terms assemble into the homotopy Jacobi relation on $V$.
The infrared contributions yield exactly the various compositions of $\mathsf{A}_{2, 1}^0$ with itself which assemble into the homotopy Jacobi relation (cf. \eqref{eq:Hty_Jacobi} for $n = 3$) on $\ker(H)$ together with $Q' \mathsf{A}^{0}_{3, 1}$.
{\begin{figure}[H]
  \centering
  \includegraphics[width=.7\columnwidth]{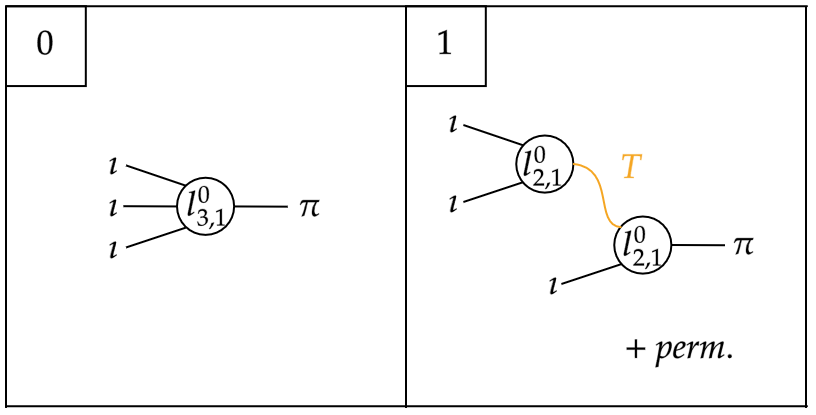}
  \caption{The relevant $0$- and $1$-cells for $\mathcal{MG}_{3, 1}^0$.}
  \label{fig:ModuliSpace_3_1_0}
\end{figure}}

Similarly for the $1$-loop relation for $\mathsf{A}_{1, 0}^{1}$ we get
\begin{align}
  Q' \mathsf{A}_{1, 0}^{1} =&~ l_{1, 0}^1 \circ Q \circ \imath - \int_{\partial (\mathcal{MG}_{1, 0}^{1,\mr{top}})} ({\PA^\infty})_{1, 0}^{1}.
\end{align}
The ultraviolet contributions from the r.h.s. (cf. \Cref{fig:ModuliSpace_1_0_1}) yield exactly the different traces over the $l_{2, 1}^0$ operation on $V$ and assemble – together with the first term – into the homotopy unimodularity equation on $V$.
Similarly the infrared contributions yield the respective traces in $\ker(H)$ and thus prove the homotopy unimodularity equation (cf. \eqref{eq:Hty_unimodularity} for $n = 1$) for the amplitudes.
{\begin{figure}[H]
  \centering
  \includegraphics[width=.7\columnwidth]{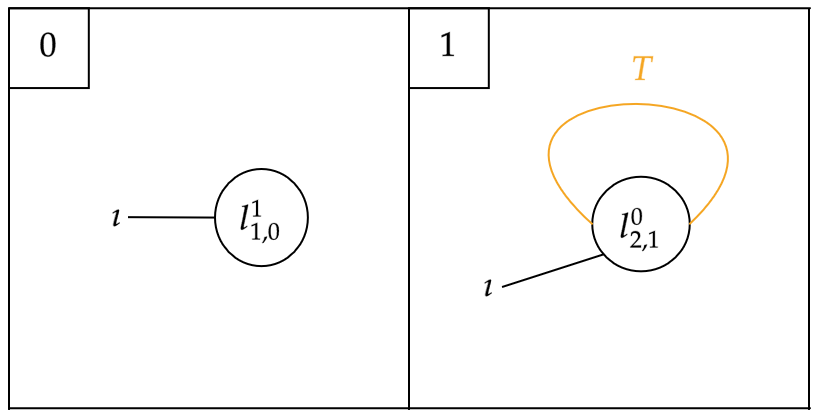}
  \caption{The relevant $0$- and $1$-cells for $\mathcal{MG}_{1, 0}^{1}$.}
  \label{fig:ModuliSpace_1_0_1}
\end{figure}}
\section{1d AKSZ coupled to 1d SUGRA}
\label{sec:1D_AKSZ}
In this section we will construct a $1$-dimensional AKSZ theory 
which provides a Lagrangian (path integral) description of the
HTQM on an interval.
In \Cref{subsec:1D_AKSZ_couplegravity} we formulate – in the BV-BFV language – a $1$-dimensional AKSZ theory coupled to 1d gravity.
\Cref{subsec:Simplification_SUGRA} demonstrates how to simplify the 1d gravity sector by constructing the BV-pushforward to bulk zero-modes and reducing the boundary space of states.
%
We conclude with \Cref{subsec:Gauge_fixing_mQME} where we discuss gauge-fixings of the resulting theory and show how to obtain HTQM on an interval.
\subsection{AKSZ Formulation}
\label{subsec:1D_AKSZ_couplegravity}
As the target for our AKSZ theory $\tb$ we take\footnote{
  Later (in \Cref{sec:1D_theory_graphs}) we will incorporate an extra factor in $\NC$, corresponding to the target $Y$ of theory $\TB$.
  In this section we are working in a simplified setting, corresponding to the factor $\Omega^\bullet(M)$ in $\FC^\TB=\Omega^\bullet(M)\otimes Y$, cf. \eqref{eq:intro_AKSZ_fields}, \eqref{eq:Omega_pt_Ht_pt}.
}
$\NC \coloneqq \Tan^*(\Tan[1](M \times \RB[1]))$ with coordinates $(\pi_i, p_i, \theta^i, x^i)$ (degrees $-1, 0, 1, 0$) associated to local coordinates $x^i$ on $M$ and $(p_\phi, p_c, \phi, c)$ (degrees $-2, -1, 2, 1$) associated to the coordinate $c$ on $\RB[1]$.
On $\NC$ we have the standard symplectic form
\begin{align}
\label{eq:omega_N}
    \omega_{\NC} &\coloneqq \omega_M + \omega_{\RB[1]} =
    \delta p_i \wedge \delta x^i
    + \delta  \pi_i \wedge \delta  \theta^i
    + \delta p_c \wedge \delta c
    + \delta p_\phi \wedge \delta \phi.
\end{align}
As the target Hamiltonian we take 
\begin{equation}
\label{eq:Theta_N}
    \Theta_{\NC} \coloneqq p_i \theta^i + p_c \phi.
\end{equation}
For the source we choose $\Tan[1]I$ (where $I$ is the unit interval), thus we get the following space of fields:
\begin{align}
  \FC^{\tb} \coloneqq \Map(\Tan[1] I, \Tan^*(\Tan[1](M \times \RB[1]))).
\end{align}
On $\FC^{\tb}$ we have lifted coordinates (the AKSZ superfields):
\begingroup
\allowdisplaybreaks
\begin{align}
  (\widetilde{x}^i)^{(0)} &= (x^i)^{(0)}_{(0)} + (p^{\dagger i})^{(1)}_{(-1)}, &&
  (\widetilde{p}_i)^{(0)} = (p_i)^{(0)}_{(0)} + (x^{\dagger}_i)^{(1)}_{(-1)}, \\
  (\widetilde{\theta}^i)^{(1)} &= (\theta^i)^{(0)}_{(1)} + (\pi^{\dagger i})^{(1)}_{(0)}, &&
  (\widetilde{\pi}_i)^{(-1)} = (\pi_i)^{(0)}_{(-1)} + (\theta^{\dagger}_i)^{(1)}_{(-2)}, \\
  \widetilde{c}^{(1)} &= (c)^{(0)}_{(1)} + (e)^{(1)}_{(0)}, &&
  \widetilde{p}_c^{(-1)} = (p_c)^{(0)}_{(-1)} + (c^{\dagger})^{(1)}_{(-2)}, \\
  \widetilde{\phi}^{(2)} &= (\phi)^{(0)}_{(2)} + (\epsilon)^{(1)}_{(1)}, &&
  \widetilde{p}_\phi^{(-2)} = (p_\phi)^{(0)}_{(-2)} + (\phi^{\dagger})^{(1)}_{(-3)}.
\end{align}
\endgroup
Here in each r.h.s. the upper index is the form-degree and the lower index is the internal degree (ghost number); in each l.h.s., the upper index is the total degree.
We get an AKSZ BV theory $(\FC^\tb, S^\tb, \Omega^\tb, Q^\tb)$, where
\begin{align}
\label{eq:Action_tb}
  S^\tb \coloneqq&~ \int_I \widetilde{p}_i \dR \widetilde{x}^i
  + \widetilde{\pi}_i \dR \widetilde{\theta}^i
  + \widetilde{p}_i \widetilde{\theta}^i
  + \widetilde{p}_c \dR \widetilde{c} + \widetilde{p}_\phi \dR \widetilde{\phi}
  + \widetilde{p}_c \widetilde{\phi} \\
  =&~ \int_I p_i \dR x^i
  + \pi_i \dR \theta^i
  + p_i \pi^{\dagger i}
  + x^\dagger_i \theta^i
  + {p_c} \dR c
  + p_\phi \dR \phi
  + c^\dagger \phi
  + \epsilon {p_c}, \\
  \Omega^\tb \coloneqq&~ \int_I \delta \widetilde{p}_i \delta \widetilde{x}^i
  + \delta \widetilde{\pi}_i \delta \widetilde{\theta}^i
  + \delta \widetilde{p}_c \delta \widetilde{c}
  + \delta \widetilde{p}_\phi \delta \widetilde{\phi},
\end{align}
and $Q^{\tb}$ is the cohomological vector field whose action on superfields is given by
\begin{alignat}{4}
    Q^{\tb} \widetilde{x}^i &= \dR \widetilde{x}^i + \widetilde{\theta}^i, \qquad
    &&Q^{\tb} \widetilde{\theta}^i = \dR \widetilde{\theta}^i, \qquad
    &&Q^{\tb} \widetilde{p}_i = \dR \widetilde{p}_i, \qquad
    &&Q^{\tb} \widetilde{\pi}_i = \dR \widetilde{\pi}_i + \widetilde{p}_i ,\\
    Q^{\tb} \widetilde{c}^i &= \dR \widetilde{c} + \widetilde{\phi}, \qquad
    &&Q^{\tb} \widetilde{\phi} = \dR \widetilde{\phi}, \qquad
    &&Q^{\tb} \widetilde{p}_c = \dR \widetilde{p}_c, \qquad
    &&Q^{\tb} \widetilde{p}_{\phi} = \dR \widetilde{p}_{\phi} + \widetilde{p}_c .
\end{alignat}
It is the Hamiltonian vector field generated by $S^{\tb}$ (modulo boundary terms, cf. \cite{Cattaneo:2012qu}).
\begin{remark}
  Adding $\Tan^*(\Tan[1] \RB[1])$ to the target corresponds to promoting the length modulus of the interval to a field and adding its 
  fermionic counterpart (superpartner) as well as the corresponding ghosts and antifields.
  This is referred to as ``coupling to topological gravity'' in \cite{Frenkel:2006fy}.
  It ensures invariance under time reparametrisations – time being the coordinate of the interval – with extra fields to preserve the $Q$-symmetry.
\end{remark}

The space of fields $\mc{F}^\tb = \mc{F}_{\sigma}\times \mc{F}_\SUGRA$ is the Cartesian product of the ``1d sigma-model sector'' $\FC_\sigma$ with target $\Tan^*(\Tan[1] M)$ and the ``supergravity (SUGRA) sector'' $\FC_\SUGRA$ with target $\Tan^*(\Tan[1] \RB[1])$.
All the structure on $\mc{F}^\tb$ introduced above splits into these two sectors.
\begin{remark}
\label{rem:hbar_equals_1}
    Note that the partition function of theory $\tb$ does \textit{not} have a parameter $\hbar$ (i.e. it is ``strongly quantum'' rather than quasiclassical).
    One reason for this will become clear in \Cref{sec:DequantLie}, where we extend $\tb$ to graphs with special canonical relations assigned to vertices.
    Here introducing a parameter $\hbar$ would correspond to a scaling of a certain prequantum connection which, in general, might not be possible.
    An explicit example where this fails is given in \Cref{subsubsec:Intertwiners_su2}, where scaling the prequantum connection would correspond to scaling simultaneously the spins describing the representations of $\su(2)$.
    
    Similarly, when in \Cref{sec:1D_theory_graphs} we extend $\tb$ to accommodate the target $Y$ of theory $\TB$, we will need to add a term to the action, which generally does not allow scaling by $\hbar$, see \Cref{extending e^(iS^t) by Hol factor}.
\end{remark}
\subsubsection{BV-BFV Analysis}
\label{subsubsec:BVBFV_1D_Sugra}
On a boundary point of $I$ with orientation $\pm$ one gets the exact BFV manifold \\
${(\FC^\partial, \omega^\partial_\pm \coloneqq \delta \alpha^\partial_\pm, Q^\partial)}$, where
\begin{align}
  \FC^\partial =&~  \Tan^*(\Tan[1] (M \times \RB[1])) \ni (\overline{\pi}_i, \overline{p}_i, \overline{\theta}^i, \overline{x}^i, \overline{p}_\phi, \overline{p}_c, \overline{\phi}, \overline{c}), \\
  \alpha^\partial_\pm =&~ \pm (\overline{p}_i \delta \overline{x}^i + \overline{\pi}_i \delta \overline{\theta}^i + \overline{p}_c \delta \overline{c} + \overline{p}_\phi \delta \overline{\phi}), \\
  %
  \intertext{with (degree $1$) BFV action}
  S^\partial_\pm =&~ \pm (\overline{p}_i \overline{\theta}^i + \overline{p}_c \overline{\phi}),
  \intertext{and associated Hamiltonian vector field w.r.t. the symplectic form $\omega^\partial_\pm$ is the cohomological vector field given by
  }
  Q^\partial \overline{x}^i =&~ \overline{\theta}^i, \qquad
  Q^\partial \overline{\theta}^i = 0, \qquad
  Q^\partial \overline{p}_i = 0, \qquad
  Q^\partial \overline{\pi}_i = \overline{p}_i, \\
  Q^\partial \overline{c}^i =&~ \overline{\phi}^i, \qquad
  Q^\partial \overline{\phi}^i = 0, \qquad
  Q^\partial \overline{p}_c = 0, \qquad
  Q^\partial \overline{p}_{\phi} = \overline{p}_c .
\end{align}

The BFV $1$-form $\alpha^\partial$ is compatible with the $\left( \ddelta{}{\overline{p}_i}, \ddelta{}{\overline{\pi}_i}, \ddelta{}{\overline{p}_c}, \ddelta{}{\overline{p}_\phi} \right)$-polarisation (vanishes on its leaves).
In this polarisation the coordinates on the space of leaves of the associated foliation are given by
$\overline{x}^i, \overline{\theta}^i, \overline{c}, \overline{\phi}$.
The BFV boundary action – using the above choice of polarisations – quantises to the coboundary operator
\begin{align}\label{eq:Omega_pt_Ht_pt}
  \Omega^\partial_\pt &= -\mr{i}
  \left( \overline{\theta}^i \ddelta{}{\overline{x}^i} + \overline{\phi} \ddelta{}{\overline{c}} \right),
  \quad \text{acting on} \quad
  \HC^\tb_\pt \coloneqq \Fun(\overline{x}^i, \overline{\theta}^i, \overline{c}, \overline{\phi}) = \Omega^\bullet(M) \otimes \Omega^\bullet(\RB[1]),
\end{align}
which is the (normalised) de Rham operator.
\subsection{Simplification of SUGRA sector}
\label{subsec:Simplification_SUGRA}
\subsubsection{BV-pushforward in the Bulk}
\label{subsubsec:BVpush_SUGRA_bulk}
Choosing the same polarisation on both boundary points, the space of fields for SUGRA fibres over the base $\BC = \Tan[1]\RB[1] \oplus \Tan[1] \RB[1]$ with coordinates $(c_\text{in}, \phi_\text{in}, c_\text{out}, \phi_\text{out})$ and fibre
\begin{align}
  \YC =
    \Omega^\bullet(I, \partial I; \underbrace{\Tan[1]\RB[1]}_{W_1})
  \oplus
    \Omega^\bullet(I; \underbrace{\Tan[-1]\RB[-1]}_{W_2})
  .
\end{align}
Taking (discontinuous) extensions of boundary fields into the bulk \cite[section 2.4]{Cattaneo:2015vsa} we write $\FC_\SUGRA = 
\BC \times \YC$.
We split $\YC = \YC' \oplus \YC''$ using the following strong deformation retraction (SDR):
\begin{center}
\begin{tikzcd}[sep=large]
  \arrow[loop left, distance=3em, start anchor={[yshift=-1ex]west}, end anchor={[yshift=1ex]west}]{}{K}
  (\Omega^\bullet(I, \partial I; W_1) \oplus \Omega^\bullet(I; W_2), \dR)
  \arrow[r, twoheadrightarrow, shift left, "p"] &
  \arrow[l, shift left, hook', "\imath"]
  (\mr{H}^\bullet(I, \partial I; W_1) \oplus \mr{H}^\bullet(I; W_2), 0)
\end{tikzcd}
\end{center}
We use $\mr{H}^\bullet(I, \partial I; W_1) \simeq \langle \dR t \rangle \otimes W_1$ and $\mr{H}^\bullet(I; W_2) \simeq \langle 1 \rangle \otimes W_2$ to define
\begin{align}
  K \colon \alpha \in \Omega^1(I, \partial I; W_1) &\longmapsto \int_0^t \alpha - t \int_0^1 \alpha, &&\beta \in \Omega^1(I; W_2) \longmapsto \int_0^1 t' \beta - \int_t^1 \beta, \\
  p \colon \alpha \in \Omega^1(I, \partial I; W_1) &\longmapsto \dR t \int_0^1 \alpha, &&\beta \in \Omega^1(I; W_2) \longmapsto 0, \\
  f \in \Omega^0(I, \partial I; W_1) &\longmapsto 0, &&g \in \Omega^0(I; W_2) \longmapsto \int_0^1 g(t) \dR t.
\end{align}
Here the inclusion $\imath$ is the trivial inclusion of cohomology elements.
This SDR gives us a splitting of the bulk fields into
\begin{align*}
  \YC &= \underbrace{\left( (\langle \dR t \rangle \otimes W_1) \oplus (\langle 1 \rangle \otimes W_2) \right)}_{\YC' \ni ((\epsilon^I, e^I), ({p_\phi^I}, {p_c^I}))}
  \medoplus \underbrace{\left( \Omega^0(I, \partial I; W_1) \oplus \Omega^0_{I}(I; W_2) \right)}_{\YC''_{\im(K)} \ni ((\phi'', c''), (p_\phi'', {p_c''}))}
  \medoplus \underbrace{\left( \Omega^1_{I}(I, \partial I; W_1) \oplus \Omega^1(I; W_2) \right)}_{\YC''_{\im(d)} \ni ((\epsilon'', e''), (\phi^{\dagger ''}, c^{\dagger ''}))}.
\end{align*}
Here we denote by $\Omega^\bullet_I$ differential forms with vanishing integral over $I$ for $1$-forms and for $0$-forms respectively after multiplication by $\dR t$.
We perform an integration over $\YC''_{\im(K)}$, setting the $\YC''_{\im(\dR)}$ component to zero to define the Lagrangian $\LC_{UV}$:
\begin{align}
  Z_{\SUGRA}[\overline{c}, \overline{\phi}, e^I, \epsilon^I, {p_\phi^I}, {p_c^I}] &= \int_{\LC_{UV}} \DC[c'', \phi'', {p_c''}, p_\phi''] \ \e^{
  \mr{i} S_{\SUGRA}} \\
  &= \e^{
  \mr{i}\left( \epsilon^I {p_c^I} - {p_c^I} (\overline{c}_\mr{out} - \overline{c}_\mr{in}) + {p_\phi^I} (\overline{\phi}_\mr{out} - \overline{\phi}_\mr{in}) \right)}.
\end{align}
$Z_{\SUGRA}$ satisfies the modified Quantum Master Equation
\begin{align}
\label{eq:mQME}\tag{mQME}
  (\Omega_{\SUGRA}^\partial +
  \Delta_{\SUGRA}') Z_{\SUGRA} = 0,
\end{align}
where
\begin{align}
  \Omega_{\SUGRA}^\partial &= -\mr{i}
  \int_{\partial I} \overline{\phi} \ddelta{}{\overline{c}}
  , \qquad \qquad
  \Delta_{\SUGRA}' = \left( \ddelta{}{\epsilon^I} \ddelta{}{p_\phi^I} + \ddelta{}{e^I} \ddelta{}{{p_c^I}} \right).
\end{align}
\subsubsection{Reduction of the Boundary Space of States}
\label{subsubsec:Reduction_SUGRA_Boundary}
To fully reduce the SUGRA sector, we want to reduce the space of boundary states to its cohomology, following \cite[section 7.4]{Cattaneo:2017tef}.
The idea is to pass to the cohomology of
\begin{align}
  \left( \HC_{\SUGRA} = \Fun_\CB(\overline{c}, \overline{\phi}) = \Omega^\bullet(\RB[1])
  , \ \
  \Omega^\partial_{\SUGRA} =
  - \mr{i}\overline{\phi} \dell{}{\overline{c}} \right),
\end{align}
%
If $Z_{\SUGRA}$ were $\Omega^\partial_{\SUGRA}$-closed, we could work with its cohomology class.
However instead, $Z_{\SUGRA}$ satisfies the \eqref{eq:mQME}.
Thus we define a strong deformation retraction
\begin{center}
\begin{tikzcd}[sep=large]
  \arrow[loop left, distance=3em, start anchor={[yshift=-1ex]west}, end anchor={[yshift=1ex]west}]{}{K^\partial} \HC_{\SUGRA} \arrow[r, twoheadrightarrow, shift left, "p^\partial"] & \arrow[l, shift left, hook', "\imath^\partial"]  \mr{H}^\bullet(\HC_{\SUGRA}) = \CB,
\end{tikzcd}
\begin{align}
  p^\partial \colon \psi(c, \phi) &\longmapsto \psi(0, 0), \\
  K^\partial \colon \psi(c, \phi) &\longmapsto c \frac{\psi(c, \phi) - \psi(c, 0)}{\phi}.
\end{align}
\end{center}
$\imath^\partial$ acts as the inclusion of constants.
Using this data, define the modified partition function
%
\begin{align}
    Z^{\Mod}_{\SUGRA} \coloneqq&~ \imath^\partial \circ p^\partial (Z) = Z - \Omega^\partial \circ K^\partial(Z) - K^\partial \circ \underbrace{\Omega^\partial(Z)}_{= -\Delta' Z} \\
    =&~ Z - \left( \Omega^\partial +
    \Delta' \right)(K^\partial(Z)).
\end{align}
%
Here $Z$, $\Omega^\partial$ and $\Delta'$ are short for $Z_\SUGRA$, $\Omega^\partial_\SUGRA$ and $\Delta'_\SUGRA$ respectively.
$Z^{\Mod}_{\SUGRA}$ is $\Omega^\partial_{\SUGRA}$-closed by construction and differs from $Z_{\SUGRA}$ only by a BV-BFV exact term.

We can thus define the \textbf{reduced partition function}
\begin{align}\label{Z^red_SUGRA}
  Z^{\red}_{\SUGRA} \coloneqq&~ p^\partial(Z^\Mod_{\SUGRA}) = [Z^\Mod_{\SUGRA}] = \e^{
  \mr{i}\; \epsilon^I {p_c^I} }.
\end{align}
\subsection{Gauge-fixing}
\label{subsec:Gauge_fixing_mQME}
\subsubsection{Good Gauge-fixing Data}
\label{subsubsec:Good_gauge_fixings}
Let $\Ghat$ be some square-zero operator of degree $-1$ acting on forms on $M$ (a ``gauge-fixing operator'' or ``non-normalised chain homotopy''),\footnote{
  A recent discussion of good gauge-fixing operators can be found in \cite[Section 2.5.1]{Mnev:2025tko}.
}
and define $\Hhat \coloneqq [\dR_M, \Ghat]$, where $\dR_M$ is the de Rham differential on 
$M$.
We assume that $\Hhat$ is diagonalisable (which implies a decomposition into a direct sum $\Omega^\bullet(M) = \ker(\Hhat)\oplus \im(\Hhat)$), with non-negative real part of spectrum.\footnote{
  For renormalisation purposes, one could additionally ask (cf. \cite[Section 3.1]{Costello:2007ei} ) that the operators $\Ghat$ and $\Hhat \coloneqq [\dR_M, \Ghat]$ are elliptic as endomorphisms of $\Omega^\bullet(M)$ and that the latter is a generalised Laplacian \cite{BGV_1992}.
}
We further assume that $\Ghat$ vanishes on $\ker(\Hhat)$.\footnote{
    If additionally $\dR_M$ vanishes on $\ker(\Hhat)$, one has a ``full'' or ``complete'' gauge-fixing.
}

One then has a \textbf{generalised Hodge decomposition}
\begin{align}
  \Omega^\bullet(M) &= \ker(\Hhat) \oplus \underbrace{\left( \im(\Ghat) \oplus \left(\im(\dR_M) \cap \im(\Hhat)\right) \right)}_{\im(\Hhat)}.
\end{align}
%
The propagator on an edge has an integral expression
\begin{align}
\label{eq:CS_Propagator}
   \int_0^\infty \e^{- T \Hhat - \dR T \Ghat}&=
   \begin{cases}
       \Ghat ([\dR_M, \Ghat])^{-1} &\text{on } \im(\Hhat) \\
       0 &\text{on } \ker(\Hhat)
  \end{cases}.
\end{align}

Later on, when assembling the theory $\tb$ on graphs as the first quantisation picture for theory $\TB$, we will have the following changes:
\begin{enumerate}[label=(\alph*)]
    \item
        $\Ghat$ will act on $\FC^\TB = \Omega^\bullet(M) \otimes Y$ instead of $\Omega^\bullet(M)$,
    \item
        we will require that $\Ghat$ is self-adjoint with respect to the pairing (the BV symplectic form $\omega^\TB$) on $\FC^\TB = \Omega^\bullet(M) \otimes Y$.
\end{enumerate}
The latter property guarantees that $\mr{im}(\Ghat) \subset \FC^\TB$ is isotropic, which is necessary to define a valid gauge-fixing of the theory $\TB$.
\subsubsection{Reproducing the Propagator}
\label{subsubsec:Reproduce_prop}
Starting with $\tb$ as a theory defined by action
\begin{align}
    S^\tb =
        \underbrace{\int_I p_i \dR x^i + \pi_i \dR \theta^i + \til{p}_i \til\theta^i}_{S_\sigma}
        + \underbrace{\int_I p_c \dR c+ p_\phi \dR \phi+ \til{p}_c \til\phi}_{S_\SUGRA}
\end{align}
on the space of fields $\mc{F}_\sigma \times \mc{F}_\SUGRA$, we integrate out SUGRA fields (retaining SUGRA zero-modes $\mc{Y}'_\SUGRA = \{(e^I, \epsilon^I, p_c^I, p_\phi^I)\}$) and set boundary values of SUGRA fields $c, \phi$ to zero as in \Cref{subsec:Simplification_SUGRA}.
This yields a theory with the space of fields $\mc{F}_\sigma \times \mc{Y}'_\SUGRA$ and action 
$S_\sigma + S^\mr{eff}_\SUGRA$, with $S^\mr{eff}_\SUGRA = \epsilon^I p_c^I$, cf. \eqref{Z^red_SUGRA}.

Define a gauge-fixing fermion of the form\footnote{
  We obtain here HTQM on an interval (cf. 
  \Cref{subsec:HTQFT}) and the gauge-fixing/reduction process is presented in two steps:
  We first simplify the SUGRA sector and then gauge-fix the rest.
  One can skip the first step and immediately gauge-fix, allowing for more general dependence of the gauge-fixing fermion on the SUGRA fields.
}
\begin{align}
\label{eq:Psi}
  \Psi \coloneqq&~ - \int_I \dR t\, e^I \, \Gcl(x^i, \theta^i, p_i, \pi_i),
\end{align}
with $\Gcl$ some degree $-1$ function on $\Tan^*\Tan[1]M$.
Consider the symplectomorphism $\Phi$ of $\mc{F}_\sigma\times \mc{Y}'_\SUGRA$ given by the flow of the Hamiltonian vector field $\{-,\Psi\}$ in time 1.
We have
\begin{equation}
    \Phi^* (S_\sigma+S^\mr{eff}_\SUGRA)=
    S_\sigma+S^\mr{eff}_\SUGRA -\int_I \dR t\, (e^I \Hcl + \epsilon^I \Gcl),
\end{equation}
where
\begin{equation}
\Hcl \coloneqq Q^\tb(\Gcl)=\left(\theta^i\frac{\partial}{\partial x^i}+p_i \frac{\partial}{\partial \pi_i}\right) \Gcl.
\end{equation}

We consider the space of fields of the 1d sigma-model $\mc{F}_\sigma$ as a trivial fibre bundle over the space of boundary conditions $\mc{B}_\sigma = \{(\overline{x}^i, \overline\theta^i)\}$ with fibre $\mc{Y}_\sigma = \{\til{x}^i_0, \til\theta^i_0, \til{p}_i, \til\pi_i\}$, where the subscript $0$ indicates that the $0$-form component of the respective field vanishes at $t = 0, 1$.

In $\mc{Y}_\sigma$ we have the ``trivial'' Lagrangian subspace $\mc{L}^0_\sigma$ defined by setting the 1-form components of superfields to zero.

Then for the gauge-fixed path integral of theory $\tb$ we have
\begin{multline}
\label{eq:Partition_function_tb_G}
  Z_G^\tb 
  = \int_{\mc{L}^0_\sigma\subset \mc{F}_\sigma} 
  e^{\mr{i} \Phi^*(S_\sigma+S_\SUGRA^\mr{eff})
  } \\
  = \int \DF \left[ x^i_0, \theta^i_0, p_i, \pi_i \right] \exp
    \Bigg(
        \mr{i}\Bigg(
            \epsilon^I{p_c^I} 
            + \int_I p_i \dR x^i_0
            + \pi_i \dR \theta^i_0
            - \dR t\, e^I \Hcl
            - \dR t\, \epsilon^I \Gcl
            + \int_{\partial I} p_i \overline{x}^i
            + \pi_i \overline{\theta}^i
        \Bigg)
    \Bigg) \\
    \in C^\infty(\mc{B}_\sigma\times \mc{Y}'_\SUGRA).
\end{multline}
%
%
%
This partition function is 
the integral kernel for the operator 
\begin{align}
\label{eq:Partitionfunction_coupled_AKSZintegrated}
  Z_G^\tb = \e^{
        - e^I \Hhat
        - \epsilon^I \Ghat
        + \mr{i} \epsilon^I p_c^I 
    } ,
\end{align}
acting on the space of fields at a point (by abuse of notations, we use the same notation for the operator and its integral kernel).
Here $G$ and $H$ are defined as $\mr{i}$ times the canonical quantisation of $\Gcl,\Hcl$: 
%
\begin{align}
\label{eq:Ghat_via_G}
  \Ghat \coloneqq&~
  \mr{i}\Gcl \left( \overline{x}^i, \overline{\theta}^i, \overline{p}_i = - \mr{i}
  \dell{}{\overline{x}^i}, \overline{\pi}_i = - \mr{i}
  \dell{}{\overline{\theta}^i} \right), \\
  \Hhat \coloneqq&~
  \mr{i}\Hcl \left( \overline{x}^i, \overline{\theta}^i, \overline{p}_i = - \mr{i}
  \dell{}{\overline{x}^i}, \overline{\pi}_i = - \mr{i}
  \dell{}{\overline{\theta}^i} \right).
\end{align}
We assume that some ordering is fixed such that $G^2=0$, $G$ is self-adjoint and $H=[\mr{d}_M,G]$.\footnote{
    One can obtain \eqref{eq:Partitionfunction_coupled_AKSZintegrated} by a perturbative evaluation of the path integral \eqref{eq:Partition_function_tb_G}.
    For that, it is convenient to modify the construction by fixing boundary conditions on $x$ and $\theta$ at $t = 1$ and on $p, \pi$ at $t = 0$, i.e., one splits fields into boundary fields (extended discontinuously into the bulk by zero -- we denote this extension by $(\cdots)^\circ$) and bulk fields:
    $x = x_\mr{out}^\circ + \hat{x}$, $\theta = \theta_\mr{out}^\circ + \hat\theta$, $p = p_\mr{in}^\circ + \hat{p}$, $\pi = \pi_\mr{in}^\circ + \hat\pi$  with $\hat{x}|_{t = 1} = \hat\theta|_{t = 1} = \hat{p}|_{t = 0} = \hat\pi|_{t = 0} = 0$.
    Then one does not have zero-modes for the bulk fields $\hat{x}, \hat\theta, \hat{p}, \hat\pi$ and one has propagators $\langle \hat{p}_i(t) \hat{x}^j(t') \rangle = - \mr{i} \delta^j_i \vartheta(t - t')$,
    $\langle \hat{\pi}_i(t) \hat{\theta}^j(t') \rangle = - \mr{i} \delta^j_i \vartheta(t - t')$.
    Then, after performing the perturbative computation, one can change the polarisation at $t = 0$ from ``$p, \pi$ fixed'' to ``$x, \theta$ fixed'' by Fourier transform, cf. \cite[Section 4.4]{Cattaneo:2015vsa}. 
}
\begin{example}[Morse--Witten gauge-fixing]
\label{ex:Morse_Witten_gaugefixing}
  Fix a metric $g$ on $M$ with Christoffel symbols $\Gamma^i_{jk}$ and a Morse function $f$ with gradient vector field $v$ and let $\varepsilon \geq 0$.
  %
  %
  %
  Then we can make the following choice of $\Gcl$:
  %
  \begin{align}\label{eq:Gcl_MorseWitten}
    \Gcl_{\varepsilon}(x^i, \theta^i, p_i, \pi_i) \coloneqq&~
      \pi_i v^i(x) 
      - \mr{i} \varepsilon \left(
        g^{ij} \pi_i p_j
        - 
        g^{ij} \pi_i \Gamma^k_{j l} \theta^l \pi_k
    \right) .
    \end{align}
    Note that the expression in brackets is the symbol of the operator $\dR^*$, see Proposition \ref{prop: d^* in loc coordinates}.
    This choice of $\Gcl$ yields
    \begin{align}
    \Hcl_{\varepsilon}(x^i, \theta^i, p_i, \pi_i) \coloneqq&~
    \left( p_i v^i(x) - \partial_i v^j(x) \pi_j \theta^i \right) \\
    -& \mr{i} \varepsilon\Big(
      g^{ij} p_i p_j
      - 2 g^{ij} \Gamma^k_{jl} p_i \theta^l \pi_k
      + \partial_m (g^{ij} \Gamma^k_{jl}) \theta^m \theta^l \pi_i \pi_k
    \Big),
  \end{align}
  %
  so that (with an appropriate choice of ordering) we have
  \begin{align}\label{eq:MorseWitten_G_H}
    \Ghat_{\varepsilon} &= \imath_v + \varepsilon \dR^*, &&
    \Hhat_{\varepsilon} = \lie{v} + \varepsilon \Delta.
  \end{align}
  %
  %
  This gauge-fixing corresponds to Morse-Witten TQM (cf. \eqref{eq:G_MorseWitten_Ex2.3}, (\ref{eq:H_MorseWitten_Ex2.3}) and \cite{Witten:1982im}, \cite[Section 2]{Frenkel:2006fy}).
  The exponent in \eqref{eq:Partition_function_tb_G} in this case is an action of Mathai--Quillen type on the interval localising in the limit $\varepsilon\ra 0$ on gradient trajectories of the Morse function.
  For $\varepsilon > 0$ one should understand the operator
  \begin{align}
    K_\varepsilon \coloneqq - \int_0^\infty \dR T \ \Ghat_{\varepsilon} \e^{- T \Hhat_{\varepsilon}},
  \end{align}
  as a smeared (or regularised) Morse chain homotopy (for a more detailed discussion of such smearings, see \cite[Section 6.5]{Kontsevich:2000yf} and \cite[Section 2.3]{Chekeres:2021ieg}):
  In the limit $\varepsilon \rightarrow 0$ (pure Morse gauge) 
  $K_\epsilon$ can be seen as an integral operator with \textit{distributional} kernel $\delta_{Y}$, where
  \begin{align}
    Y \coloneqq \bigcup_{t > 0} Y_t \hspace{.5em} \subset M \times M, \qquad Y_t \coloneqq \{ (p, q) \in M \times M \mid p = \mathrm{Flow}_{v, t}(q) \},
  \end{align}
  and $\ker H_\epsilon$ is given by Morse cochains represented by delta-forms on unstable manifolds of critical points of $f$.
  For $\varepsilon \rightarrow \infty$ (Lorenz gauge), it is instead the Hodge chain homotopy, constructed in terms of the heat flow operator on $M$.
  %
  Thus it is natural to see the case $\varepsilon > 0$ as an interpolation between these two cases.
\end{example}
\begin{remark}[Covariant Momentum]
  If $\Gamma$ is non-trivial – to ensure that all coordinates transform tensorially – we should switch from the canonical cotangent coordinate $p_i$ to the covariant coordinate $P_i \coloneqq p_i - \Gamma^{j}_{ik} \theta^k \pi_j$.
  Note that $P_i$ (using the standard, derivatives-to-the-right ordering) quantises to the covariant derivative $- \mr{i} \dell{}{x^i} + \mr{i} \Gamma^{j}_{i k} \theta^k \dell{}{\theta^j}$.
  In terms of $P_i$, $\Gcl$ and $\Hcl$ take the following form:
  \begin{align}\label{eq:Gcl_MorseWittenCovar}
    \Gcl_{\varepsilon}(x^i, \theta^i, P_i, \pi_i) \coloneqq&~
      \pi_i v^i(x) 
      - \mr{i} \varepsilon \left(
        g^{ij}(x) \pi_i P_j
    \right) , \\
    \Hcl_{\varepsilon}(x^i, \theta^i, P_i, \pi_i) \coloneqq&~
    \left( P_i v^i(x) - \nabla_i v^j(x) \pi_j \theta^i \right) \\
    -& \mr{i} \varepsilon\Big(
      g^{ij}(x) P_i P_j
      + \frac12 g^{ij}(x) \tensor{R}{^{m}_{jkl}} \theta^k \theta^l \pi_i \pi_m
    \Big),
  \end{align}
\end{remark}
\subsubsection{HTQM and mQME}
\label{subsubsec:HTQFT_mQME}
The boundary operator and Laplacian for our gauge-fixed theory are
\begin{align}
\label{eq:ResidualBoundaryOperator}
    \Omega^\partial &\coloneqq -\mr{i}
    \overline{\theta}^i \dell{}{\overline{x}^i},
    \qquad \qquad
    \Delta = \dell{}{\epsilon^I} \dell{}{{p_\phi^I}} + \dell{}{e^I} \dell{}{{p_c^I}}.
\end{align}
Comparing with \eqref{eq:Partitionfunction_coupled_AKSZintegrated} we can see that
\begin{align}
\label{eq:qME_fullyReducedTheory}
  \left( - \mr{i}
  \Delta -
  \mr{i} \Omega^\partial \right) Z_G
  =
  -\left( \epsilon^I \dell{}{e^I} + [\dR_M, -] \right) Z_G = 0.
\end{align}
Here we denote 
$\mr{i} \Omega^\partial = [\dR_M, -]$ using $\Fun(x^i, \theta^i) \simeq \Omega^\bullet(M)$.
In other words, the \eqref{eq:mQME} implies the \eqref{eq:HTQM_closed} equation of HTQM.
%
%
%
\section{Vertex Structure}
\label{sec:DequantLie}
In order to construct the theory $\tb$ on graphs as a classical field theory (or to describe the path integral for $\tb$), we need to know what it assigns to vertices of graphs.
In \Cref{subsec:Generalisation_Linfty} we propose an ansatz for a classical object $\nu$ that is mapped by geometric quantisation to the cyclic $\LI$-algebra structure on $Y$ (the AKSZ target of theory $\TB$) arising from the Taylor expansion of $\Theta_Y$.\footnote{
    Later, in Section 5, when putting together the theory $\tb$, $Y$ will be replaced by $\Omega^\bullet(M) \otimes Y$ and $\nu$ will be extended appropriately.
}
We will call such $\nu$ a ``dequantisation'' of $Y$.
Then in \Cref{subsec:Lie_bracket_from_Lagrangian} we discuss examples corresponding to $Y$ being a Lie algebra (with the case of $\su(2)$ discussed in detail).

The paradigm proposed in this section is to be taken as tentative and possibly missing details.
A desired dictionary of geometric quantisation, translating objects of symplectic geometry to quantum objects and underlying the following considerations, can be summed up in the following table (see e.g. \cite{Bates:1997kc}):
\setlength{\arrayrulewidth}{0.3mm}
\setlength{\tabcolsep}{5pt}
\renewcommand{\arraystretch}{1.5}
\begin{center}
\begin{tabular}{|p{11cm}|p{3cm}|}
  \hline
    \textbf{Symplectic}
    & \textbf{Quantum} \\
  \hline
  \hline
    Symplectic manifold $(\Phi, \omega)$ ($+$ prequantum line bundle $\LC$, polarisation $\PC$)
    & Hilbert space $\HC$ \\
  \hline
    $\sqcup, \times,  \overline{(-)}$
    & $\oplus, \otimes, (-)^*$ \\
  \hline
    Bohr--Sommerfeld Lagrangian $L \subset \Phi$ 
    (+ horizontal section $s$ of $\mc{L}|_L$)
    & State $\psi_{L} \in \HC$ \\
  \hline
\end{tabular}
\end{center}
\begin{remark}[Ground field: $\RB$ or $\CB$]
\label{rem: ground field R or C}
    Throughout the paper, $Y$, $\gf$, $\FC^\TB$ can be considered  as (graded) vector spaces either over $\RB$ or over $\CB$. If they are considered as vector spaces over $\RB$, then we should be saying that $\nu$ quantises to the \emph{complexification} of $Y$, that the space of states of $\tb$ for a point is the \emph{complexification} of $\FC^\TB$ etc. (We will be sloppy though and not spell out this complexification step.)
    
    Alternatively, one may consider $Y$, $\gf$, $\FC^\TB$ as complex vector spaces. Then the complexification step is not needed, but when writing the path integral in the theory $\TB$, one should choose a real contour of integration.
\end{remark}
\subsection{Cyclic \texorpdfstring{$\LI$}{L∞}-algebras in Weinstein's Symplectic Category – a Tentative Definition}
\label{subsec:Generalisation_Linfty}
A. Weinstein \cite{weinstein2009symplecticcategories} proposed the definition of a category\footnote{
    The composition is defined under certain transversality assumption, so it is a category with partially defined composition.
}
with objects being symplectic manifolds $(\Phi, \omega)$ and 
morphisms $L \colon \Phi_1 \rel{} \Phi_2$ being canonical relations – Lagrangian submanifolds $L \subset \overline{\Phi}_1 \times \Phi_2$.
The overline denotes the symplectic dual (changing the sign of symplectic structure).
Canonical relations are composed as set-theoretic relations:
For two relations $\Phi_1 \rel{L_I} \Phi_2$, $\Phi_2 \rel{L_{II}} \Phi_3$, the composition is $\Phi_1 \rel{L} \Phi_3$ with
\begin{align*}
  L = L_{II} \circ L_I \coloneqq \{ (x,z) \in \overline{\Phi}_1 \times \Phi_3 \; | \; \exists y \in \Phi_2 \; \mr{s.t.} \; (x,y) \in L_I, \; (y,z) \in L_{II} \}.
\end{align*}
This category has monoidal structure (Cartesian products), with monoidal unit $\mr{pt}$ (the point), and has duals (symplectic duals).

We consider a variant of this category, which we denote $\mtt{Symp}^+$, where an object is a symplectic manifold equipped with a prequantum line bundle and a polarisation (we suppress the extra data in notation) and a morphism $L \colon \Phi_1 \rel{} \Phi_2$ is 
a quantisable (Bohr--Sommerfeld) Lagrangian in $\overline{\Phi}_1 \times \Phi_2$ equipped with a horizontal section of the prequantum line bundle over it (in particular, morphisms $\Phi_1 \rel{} \Phi_2$ form a $\CB$-vector space).\footnote{
    We are suppressing details related to the metaplectic/half-form correction (except in footnote \ref{footnote: half-forms} and Remark \ref{rem: quantizing Lagrangians in a real fibrationg polarization} below).
}\footnote{
    Monoidal structure and duals extend to $\mtt{Symp}^+$ in an obvious way. E.g., Cartesian product of symplectic manifolds is accompanied by the external tensor product of the prequantum line bundles. Symplectic dualisation of a symplectic manifold is accompanied by linear dualisation of the prequantum line bundle.
}

Changes of polarisation can be considered as a special class of automorphisms in $\mr{Symp}^+$. 

We allow the prequantum line bundle to be graded (concentrated in one degree).
For $\Phi$ an object, $\Phi^{[k]}$ will mean $\Phi$ with degree of the prequantum line bundle shifted down by $k$.\footnote{
  The idea is that if the geometric quantisation of $\Phi$ is a vector space $\mc{H}$, then the geometric quantisation of $\Phi^{[k]}$ – the space of $\mc{P}$-horizontal sections of the shifted prequantum bundle –  is $\mc{H}[k]$ – the original vector space shifted down by $k$.
}
In the first approximation these shifts can be ignored.

Geometric quantisation can be thought of as a monoidal functor
\begin{align}
\label{eq:Geom_quant_functor}
    \QC\colon \mtt{Symp}^+ \lra \mtt{grVect}_{\CB}
\end{align}
from $\mtt{Symp}^+$ to the the category of graded vector spaces over $\CB$;
$\QC$ is $\CB$-linear on morphisms.
The functor $\QC$ maps an object $\Phi$ -- a symplectic manifold with extra data -- to a vector space $\HC$, the ``space of states''.
It maps a morphism $L \colon \Phi_1 \rel{} \Phi_2$ to a linear map $\psi \colon \HC_1 \ra \HC_2$ – the ``state'' or ``evolution operator'' or ``partition function''.
\begin{remark}
    We sketch briefly how $\QC$ acts on objects and morphisms of $\mr{Symp}^+$ and refer to  \cite{Bates:1997kc}, \cite{woodhouse1992geometric} for details. 

    A symplectic manifold $(\Phi,\omega)$ equipped with prequantum line bundle $\LC$, with unitary connection $\nabla_\LC$ of curvature $-i\omega$, and equipped with a polarisation -- a Lagrangian subbundle $\PC\subset T_\CB\Phi$ -- is mapped by $\QC$ to the space of polarised ($\PC$-horizontal) sections of $\LC$
    \begin{equation} 
        \mc{H} \coloneqq \Gamma_{\PC\textrm{-hor}}(\Phi,\LC) 
    \end{equation}
    -- the space of sections of $\LC$ annihilated by covariant derivatives in the direction of vectors in $\PC$.
    The $L^2$ inner product on polarised sections is given by
    \begin{equation}\label{inner product on H}
        \langle \psi_1,\psi_2 \rangle \coloneqq \int_\Phi \mu \overline{\psi}_1\psi_2
    \end{equation}
    with $\mu \coloneqq \frac{\omega^{\wedge l}}{l!}$ the symplectic volume form, $l \coloneqq \frac12\dim \Phi$.\footnote{
        \label{footnote: half-forms}
        Alternatively, one defines states as polarised sections of $\LC$ tensored with the line bundle of half-forms on $\Phi$ (associated with a choice of metaplectic structure) -- in this approach the factor $\mu$ in (\ref{inner product on H} is not needed and the inner product is given by a canonical integration of a density on $\Phi$.
    }
    
    Suppose we are given a Lagrangian submanifold $L\subset \Phi$ and a $\nabla_\LC$-horizontal section $s$ of the prequantum line bundle $\LC$ restricted to $L$ (the existence of such a nonvanishing section implies the Bohr--Sommerfeld condition that $\nabla_\LC|_L$ has trivial holonomy).\footnote{
        In the literature (see e.g. \cite{Bates:1997kc}), it is common not to impose $\nabla_\LC$-horizontality condition on the section $s$ when defining the geometric quantisation of a pair $(L,s)$ -- and then one also does not need the Bohr--Sommerfeld condition on $L$.
        We choose to impose $\nabla_\LC$-horizontality so that for $L$ connected there is a unique up to normalisation choice of a $\nabla_\LC$-horizontal section $s$.
    } Then $\QC$ maps $(L,s)$ to the state
    \begin{equation}
       \psi_L= \sum_k \psi_k \int_\Phi \mu \overline{\psi}_k \til{s}, 
    \end{equation}
    where $\{\psi_k\}$ is an orthonormal basis in the space of states and $\til{s}=i_*s$ is the distributional section of $\LC$ on $\Phi$ supported on $L$ obtained as the pushforward of $s$ along the inclusion $i\colon L\hookrightarrow \Phi$. 
    
    Put another way:
    the space of states $\mc{H}$ is a subspace in the space $\Gamma(\Phi,\LC)$ of non-polarised sections, equipped with inner product given by the formula (\ref{inner product on H}).
    Denote the orthogonal projection from $\Gamma(\Phi,\LC)$ to $\mc{H}$ by $\PB$.
    Then the quantisation of $(L,s)$ is $\PB(\til{s})$, where $\til{s}$ is understood as a distributional element in $\Gamma(\Phi,\LC)$.
\end{remark}
\begin{remark}[Quantisation of Lagrangians in a Real Fibrating Polarisation]
\label{rem: quantizing Lagrangians in a real fibrationg polarization}
  Following \cite[section 7]{Bates:1997kc} one can think of the ``quantisation'' of Lagrangian submanifolds as follows:
  
  Let $\Phi$ be some metaplectic manifold with prequantum line bundle $\LC$.
  Then – given a real fibrating polarisation $\PC$ – to a pair $(L, s)$, where $L$ is a Lagrangian submanifold of $\Phi$ intersecting the leaves of $\PC$ transversally and in at most one point and $s \in \Gamma(L, \LC \otimes |\Lambda|^{\frac12}{L})$, one can associate a section $\widetilde{s} \in \Gamma(\Phi,\LC \otimes \Lambda^{ - \frac12} \PC)$.
  This section is covariantly constant on each leaf and vanishes on leaves which $L$ does not intersect.
  In this sense, it is a \textit{distributional} element of the Hilbert space associated to $\PC$.
  The pair $(L, s)$ is called a \textbf{semi-classical state}.
\end{remark}
\begin{definition}[A Tentative Definition]
\label{def:Min_cyc_Linfty_Symp}
  Fix $N \in \NB$ (the ``degree'').
  A \textbf{minimal}\footnote{
    I.e. without a unary differential.
  } \textbf{cyclic $\LI$-algebra in the category $\mtt{Symp}^+$} (or ``$\cLIsymp$-algebra'' for short) is the following collection of data:
  \begin{enumerate}[label=(\roman*)]
      \item
        An object $\nu\in \mr{Ob}(\mtt{Symp}^+)$.
      \item
        A morphism $\rho \colon \nu \times \nu \rel{} \mr{pt}^{[N-1]}$, or equivalently $\rho \colon \nu \rel{} \overline \nu^{[N-1]}$.
      \item
        A sequence of morphisms $L_n \colon \nu^{\times n} \rel{} \mr{pt}^{[N]}$ with $n\geq 3$.
  \end{enumerate}
  This data is subject to the following axioms:
  \begin{itemize}
    \item
      $\rho$ and $L_n$ are symmetric w.r.t. permutation of $\nu$ factors.
    \item
      $\LI$-algebra relations are satisfied:
      \begin{align}
      \label{eq:Linfty_rel_Symp}
        \sum_{\sigma \in \mathrm{S}_n} \sum_{r+s = n} \frac{1}{r! s!} L_{r+1,1} \circ (\id^{\times r} \times L_{s,1}) \circ \sigma^{\text{rel}} = 0
        , \qquad n \in \NB.
      \end{align}
      %
      Here $L_{p,q} \colon \nu^{\times q} \rel{} \nu^{[N-q(N-1)]}$ is the composition
      $\nu^{\times p} \rel{L_{p+q}} (\overline{\nu}^{\times q})^{[N]} \rel{(\rho^{-1})^{\times q}} \nu^{[N-q(N-1)]} $.
      Also, $\sigma^{\text{rel}} \colon \nu^{\times n} \rel{} \nu^{\times n}$ denotes the canonical relation permuting inputs according to $\sigma$.
  \end{itemize}
  Relations \eqref{eq:Linfty_rel_Symp} and symmetry properties of $\rho$ and $L_n$ are satisfied modulo the kernel of the geometric quantisation functor \eqref{eq:Geom_quant_functor}.
\end{definition}

Geometric quantisation $\QC$ maps $\nu$ to a graded vector space $Y$ equipped with cyclic $\LI$-algebra structure given by the pairing $\pairing_Y = \QC(\rho) \in \mr{Hom}_{N-1}(Y \otimes Y, \CB)$ and cyclic operations $c_n = \QC(L_n) \in \mr{Hom}_N(Y^{\otimes n}, \CB)$.

We will eventually want $Y$ to be the target of theory $\TB$ and $c_n$ to be the Taylor coefficients of $\Theta_Y$, cf. \eqref{eq:ST_AKSZ_action}.
\begin{remark}
\label{rem:non_minimal_cLinftysymp_algebra}
  One can extend \Cref{def:Min_cyc_Linfty_Symp} to a non-minimal cyclic $\LI$-algebra in $\mtt{Symp}^+$ as follows.
  We want $\nu$ to be a \emph{differential graded} symplectic manifold (equipped with prequantum line bundle and polarisation), with a degree $+1$ cohomological vector field $q$
  – the Hamiltonian vector field of a function $\vartheta \in \Ci(\nu)_{1}$ –
  and also equipped with canonical relations $\rho$
  and $L_n$ as in \Cref{def:Min_cyc_Linfty_Symp}.
  We again impose $\LI$-relations \eqref{eq:Linfty_rel_Symp} where $L_{1,1}$ is understood as
  \begin{equation}
      L_{1,1} = \left. \frac{\dR}{\dR \alpha} \right|_{\alpha=0} \mr{graph}(\mr{Flow}_\alpha(q)) \colon \nu \rel{} \nu^{[1]}.
  \end{equation}
  Here $\alpha$ is an infinitesimal parameter of degree $-1$.
  We also require that $\rho$ is compatible with dg manifold structure on $\nu$. 

  Quantisation of this structure is the cochain complex $\QC(\nu)$ with differential $Q = l_1 = \QC(\vartheta)$ and a sequence of operations $l_n = \QC(L_{n, 1})$, $n \geq 2$, such that the full sequence $Q = l_1, l_2, l_3, \ldots$ satisfies $\LI$-algebra relations and is compatible with the pairing $\pairing = \QC(\rho)$.
  I.e., the quantisation yields a non-minimal cyclic $\LI$-algebra with cyclic operations $c_{n + 1} = \langle -, l_n(-, \ldots, -) \rangle$, $n \geq 1$.
\end{remark}
\begin{example}[Dequantisation of $\gl_N$]
\label{ex:dequantisation_of_glN}
  Consider the Lie algebra $\gl_N 
  $ with $N \geq 1$.
  It can be seen as a quantisation of the zero-dimensional symplectic manifold $\nu = \bigsqcup_{a, b = 1}^N \mr{pt}_{ab}$ – a collection of $N^2$ points labelled by pairs $(a, b)$, with zero symplectic structure\footnote{
    Indeed, the quantisation of a point is $\CB$.
    The quantisation of a collection of points labelled by pairs $(a, b)$ is the Cartesian product $\prod_{a, b} \CB_{ab}$ of copies of $\CB$ labelled by pairs $(a, b)$.
    As a vector space, it is obviously isomorphic to $\gl(N, \CB)$.
} – equipped with the following structure:
  \begin{itemize}
      \item
        $\rho \colon \mr{pt}_{ab} \mapsto \mr{pt}_{ba}$ and
      \item
        $L_3 \coloneqq \sum\limits_{a, b, c = 1}^N \left( \pt_{ab} \times \pt_{bc} \times \pt_{ca} - \pt_{ba} \times \pt_{cb} \times \pt_{ac} \right) $ – a chain of Lagrangians in $\overline{\nu}^{\times 3}$.
  \end{itemize}
\end{example}
\begin{remark}
  The ideas of this section are tentative in nature and making them precise is deferred to future work.
  Nevertheless, the authors think it might be beneficial to compare to the works of Biran--Cornea \cite{biran2012lagrangiancobordismi, biran2018lagrangiancobordismfukayacategories} and Wehrheim--Woodward \cite{Wehrheim_2010}, in particular \cite[section 1.2]{biran2018lagrangiancobordismfukayacategories} and \cite[section 6.1, 6.3]{Wehrheim_2010}.
\end{remark}
\subsection{
Example: Dequantisation of
Lie Algebras via Coadjoint Orbits
}
\label{subsec:Lie_bracket_from_Lagrangian}
%
Let $\gf$ be a simple compact 
Lie algebra.
From Kirillov's orbit method \cite{Kirillov_1962, Kirillov:1999qgl, Hashimoto:1991pf} we find that the ``dequantisation'' of (in the sense of a classical object whose geometric quantisation is) one of its irreducible representations is given by some coadjoint orbit $\OC$.\footnote{
  For reducible representations one will have to allow for $\OC$ to be a disjoint union of coadjoint orbits.
  In particular for semi-simple Lie algebras the adjoint representation might be reducible.
}

We would like to find an object whose geometric quantisation corresponds to an intertwiner of representations, e.g. the Lie bracket.
Let $R_1, R_2, R_3$ be three representations of $\gf$, with corresponding coadjoint orbits $\OC_i$.
The triple product
\begin{align}
  \overline{\OC}_1 \times \overline{\OC}_2 \times \OC_3,
\end{align}
corresponds via geometric quantisation to $(R_1)^* \otimes (R_2)^* \otimes R_3 = \Hom(R_1 \otimes R_2, R_3)$.
Thus it is a natural place to look for Lagrangian subspaces that quantise to intertwiners.

Take the Lie bracket, seen as an element of $\Hom_G(\gf \otimes \gf, \gf)$, where by $\gf$ we also denote the adjoint representation.
The following theorem tells us that, for finite-dimensional 
simple Lie algebras, any element quantising to a non-trivial element of $\Hom_G(\gf \otimes \gf, \gf)$ should correspond (up to normalisation) to the Lie bracket:
\begin{theorem}[{\cite{King_Wybourne_1996}}]
\label{theo:Complex_Jacobi_unique}
  Let $\gf$ be a 
  simple compact Lie algebra.
  The space ${\Hom_G(\gf \otimes \gf, \gf)}$ of intertwiners is $1$-dimensional and spanned by the Lie bracket except for the cases $
  {\su(N\geq 3)}
  $, where it is $2$-dimensional, with the additional map being
  \begin{align}
    (\alpha, \beta) &\longmapsto \alpha \beta + \beta \alpha - \frac{1}{N} \tr
    (\alpha \beta + \beta \alpha) \mathds{1}_N.
  \end{align}
\end{theorem}
\noindent Inspired by this, we make two key observations:
\begin{enumerate}[label=\Roman*)]
  \item
    \textbf{A $G$-invariant Lagrangian is expected to quantise to an intertwiner of representations:}
    If $\OC_i$ are the coadjoint orbits of $\gf$ quantising to representations $R_i$, $i=1,2,3$, then it is natural to expect that the object quantising to an intertwiner -- a $G$-invariant map $R_1\otimes R_2\ra R_3$ -- is a \emph{$G$-invariant} Lagrangian ${L} \subset \overline{\OC}_1 \times \overline{\OC}_2 \times \OC_3$. 
    Moreover, it is natural to expect such $L$ to be contained in $\mu^{-1}(0)$, where
    \begin{align}\label{moment map}
        \mu \colon \overline{\OC}_1 \times \overline{\OC}_2 \times \OC_3 &\longhookright \gf^* \\
        (a, b, c) &\longmapsto c - (a + b),
    \end{align}
    is the sum of moment maps associated to the product of coadjoint orbits.\footnote{
      Recall that the moment map for a single coadjoint orbit is the inclusion $\OC \hookright \gf^*$.
    }
    Then $L$ can be pushed forward along the Marsden--Weinstein symplectic reduction to a Lagrangian 
    $\underline{L}\colon=L/G\subset \mu^{-1}(0)/G$.
    %
    By the ``$[\QC, \RC]$'' (quantisation commutes with reduction) paradigm of Guillemin--Sternberg \cite{Guillemin1982}, the space $\mu^{-1}(0) / G$ quantises to the space of intertwiners $\Hom_G(R_1 \otimes R_2,R_3)$ and $L$ or $\underline{L}$ quantises to a particular intertwiner:
    \begin{center}
    \begin{tikzcd}[sep=normal]
        L\subset \overline{\OC}_1\times \overline{\OC}_2\times \OC_3 \arrow[d, "\textrm{symp. red}"'] \arrow[r, "\QC"]
        & \QC(L) \arrow[d, phantom, sloped, "="] \arrow[r, phantom, "\in"]
        & \mr{Hom}(R_1\otimes R_2,R_3) \arrow[d, phantom, sloped, "\supset"] \\
        \underline{L} \subset \mu^{-1}(0) / G \arrow[r, "\QC"] &
        \QC(\underline{L}) \arrow[r, phantom, "\in"] &\mr{Hom}_G(R_1\otimes R_2,R_3)
    \end{tikzcd}
    \end{center}
    %
    %
    %
  \item
    \textbf{The Lagrangian corresponding to the Lie bracket on $\gf$ is expected (for a class of Lie algebras $\gf$) to be $\mu^{-1}(0)$:}
    Consider the case when $\OC_1 = \OC_2 = \OC_3 = \OC$ is the coadjoint orbit corresponding to the adjoint representation and let $\gf$ be a simple 
    Lie algebra with 
    \begin{equation}\label{g neq sl(geq 3)}
     \gf \neq \su(N \geq 3).
     \end{equation}
    Then, by \Cref{theo:Complex_Jacobi_unique}, the space of intertwiners $\Hom_G(\gf \otimes \gf, \gf)$ is 1-dimensional and is spanned by the Lie bracket. As the natural object, quantising to $\CB$ is the point, we are lead to the following.
    \begin{conjecture}\label{conj: O^3//G=pt}
        For $\gf$ a simple Lie algebra satisfying (\ref{g neq sl(geq 3)}) and $O$ the coadjoint orbit corresponding to the adjoint representation, the space $\mu^{-1}(0)/G$ (the symplectic reduction of $\overline\OC \times \overline\OC \times \OC$ as a Hamiltonian $G$-space) is a point.
    \end{conjecture}
    %
    %
    As an immediate corollary, there is a unique Lagrangian $\underline{L}=\mr{pt} \subset \mu^{-1}(0) / G$, which is the symplectic reduction of the Lagrangian $L = \mu^{-1}(0)$.
    For instance $\gf = \su(2)$ is an example, see below for details.

    For more a general $\gf$, 
    one needs to find an appropriate $\underline{L} \subset \mu^{-1}(0) / G$ quantising to the Jacobi bracket.
    Then the relevant $L$ is the preimage of $\underline{L}$ under the quotient map $\mu^{-1}(0) \ra \mu^{-1}(0) / G$.

    We refer to Appendix \ref{Appendix: dimension counts for Conjecture} for extra evidence for Conjecture \ref{conj: O^3//G=pt} from non-positivity of the virtual dimension of the symplectic reduction.
    
\end{enumerate}

\begin{example}[Dequantisation of the Killing Form]
    Let $\gf$ be a simple Lie algebra, seen as a quadratic Lie algebra, equipped with an invariant pairing given by the Killing form $\pairing_\g$, and let $\OC$ be the coadjoint orbit which quantises to the adjoint representation of $\gf$. Then the ``dequantisation'' of $\gf$ consists of $\OC$ equipped with a $G$-invariant Lagrangian $L_3 \subset \overline{\OC}^{\times 3}$ quantising to (the cyclic version of) the Lie bracket on $\g$ and the $G$-invariant Lagrangian $\rho\subset \overline{\OC}^{\times 2}$ given by pairs of antipodal points $\{(x,-x)\,|\,x\in \OC\}$ quantising to the Killing form $\pairing_\g$.\footnote{
        Note that for $\gf$ simple, this $\rho$ is the unique $G$-invariant Lagrangian in $\overline\OC^{\times 2}$.
    }
    Equivalently, $\rho$ is the kernel of the moment map $\overline\OC\times \overline\OC\ra \gf^*$.
\end{example}
\begin{example}[Double of a Lie Algebra]
    \label{example: dequantization of g+g^*}
    Consider a Lie algebra $\gf$ whose adjoint representation is the quantisation of a coadjoint orbit $\OC$, with $L_\g\subset \overline{\OC}\times \overline{\OC}\times\OC$ the Lagrangian quantising to the Lie bracket. Consider the semidirect sum of $\gf$ and its coadjoint module, $\mathfrak{h}=\gf\oplus \gf^*$.
    Then $\mathfrak{h}$ can be seen as quantisation of the disjoint union of the orbit and its symplectic dual $\nu_\mathfrak{h}=\underbrace{\OC}_{\nu_A}\sqcup \underbrace{\overline\OC}_{\nu_B}$ equipped with $L_\mathfrak{h}\subset \overline{\nu}_A\times \overline{\nu}_A\times \overline{\nu}_B\subset \overline\nu_\mathfrak{h}^{\times 3}$ which corresponds to $L_\gf$ under the identification $\overline{\nu}_A\times \overline{\nu}_A\times \overline{\nu}_B\cong \overline{\OC}\times \overline{\OC}\times\OC$.
    The datum $\rho$ is given by
    \begin{align*}
        \rho=\underbrace{\Delta}_{\subset\; \overline\nu_B\times \overline\nu_A\cong\OC\times \overline\OC}\sqcup \underbrace{\overline\Delta}_{\subset\; \overline\nu_A\times \overline\nu_B\cong\overline\OC\times \OC}\subset \overline\nu_\mathfrak{h}\times \overline\nu_\mathfrak{h}
    \end{align*}
    where $\Delta,\overline\Delta$ are the diagonals in $\OC\times \overline\OC$ and $\overline\OC\times \OC$, respectively.
    This $\rho$ quantises to the pairing on $\mathfrak{h}$ given by the canonical pairing of $\gf$ and $\gf^*$.
    
    This example trivially extends to the case when the coadjoint module $\g^*$ is replaced by its shifted copy $\g^*[k]$ for any $k$; then $\nu_B$ should be shifted by $k$ as well.
    
    This example is pertinent 
    when considering $\TB$ a $BF$ theory as an AKSZ theory with target $\Tan^*[m-1](\gf[1]) = \gf[1] \oplus \gf^*[m-2] 
    $, where $m$ is the dimension of the source $M$, see Example \ref{example: BF via coad orbits} below. 
    %
\end{example}
\subsubsection{Example: Intertwiners for \texorpdfstring{$\su(2)$}{su(2)}}
\label{subsubsec:Intertwiners_su2}
Quantisable coadjoint orbits of $\su(2)$ can be labelled by $j \geq 0$ and are spheres of radius $\frac{j}{2}$.
Since $\mu^{-1}(0)$ is coisotropic and of dimension $3$, it is itself a Lagrangian submanifold.
Further, since $\mu^{-1}(0) / \SU(2)$ is just a point, it is the unique $G$-invariant Lagrangian in the triple product.
Denote this Lagrangian by $L_{\mathrm{W}}$ (\textit{Wigner Lagrangian}\footnote{
  A $6$-dimensional relative of this manifold (as 3-torus bundle over $\SO(3)$ and an orbit of an $\SU(2) \times T^3$ action on the Schwinger phase space $(\CB^2)^{\times 3}$) has been considered before in \cite{Aquilanti:2007oga} in the context of semi-classical formulations of the Wigner $3j$-symbols and dubbed the \textit{Wigner manifold}.
}).

A basis for the Hilbert space $\HC_j$ associated to $\OC_j$ (equipped with holomorphic polarisation) is given for $m \in [-\frac{j}{2}, \frac{j}{2}] \cap \ZB$ by
\begin{align}
  \ket{j, m} \coloneqq \sqrt{\frac{(2j)!}{(j+m)!(j-m)!}} z^{j + m},
\end{align}

In this basis, the unique (up to scaling) state in $\HC_{j_1} \otimes \HC_{j_2} \otimes \HC_{j_3}$ that is $\SU(2)$-invariant and a $0$-eigenstate of the total angular momentum is:
\begin{align}
\label{eq:Invariant_state_Linv}
  \ket{L_{\mathrm{W}}} &\coloneqq \sum_{m_1, m_2, m_3}
  \begin{pmatrix}
    j_1 & j_2 & j_3 \\
    m_1 & m_2 & m_3
  \end{pmatrix}
  \ket{j_1, m_1} \ket{j_2, m_2} \ket{j_3, m_3} \\
  &= C_{\{j_i\}} (z_1 - z_2)^{j_1 + j_2 - j_3} (z_2 - z_3)^{j_2 + j_3 - j_1}  (z_3 - z_1)^{j_3 + j_1 - j_2}.
\end{align}
Here
$\left(
  \begin{smallmatrix}
    j_1 & j_2 & j_3 \\
    m_1 & m_2 & m_3
  \end{smallmatrix}
\right)$
denotes the Wigner $3j$-symbols and $C_{\{j_i\}}$ is a normalising constant associated to the $j_i$.

Above we use that $\su(2)$ is self-dual, hence intertwiners of representations $R_{j_1} \otimes R_{j_2} \otimes R_{j_3}$ and the $1$-dimensional (spin $0$) representation are equivalent to intertwiners $R_{j_1} \otimes R_{j_2} \rightarrow R_{j_3}$.

Note that $\ket{L_{\mathrm{W}}}$ gives the following integral formula for the Wigner $3j$-symbols using the inner product on the product Hilbert space:
\begin{align}
  \frac{N(\{j_i\}, \{m_i\})}{C_{\{j_i\}} \pi^3} \int_{\CB^3} \dR {\bf{z}} \ \prod_{i=1}^3 \left( \frac{1}{(1 + z_i \overline{z_i})^{2 j_i + 2}} \right) \ \overline{z_1}^{j_1 + m_1} \overline{z_2}^{j_2 + m_2} \overline{z_3}^{j_3 + m_3} \ket{L_{\mathrm{W}}}.
\end{align}
Here
\begin{align}
    N(\{j_i\}, \{m_i\}) \coloneqq \prod_{i = 1}^3 \sqrt{\frac{(2 j_i)!}{(j_i + m_i)! (j_i - m_i)!}} .
\end{align}
\begin{example}\label{example: L_W for su(2)}
  Let $j_1 = j_2 = j_3 = 1$ and denote by $\OC \coloneqq \OC_1$ the coadjoint orbit corresponding to the adjoint representation.
  There is an immediate geometric interpretation of the Lagrangian
  \begin{align}
    L_{\mathrm{W}} \coloneqq \mu^{-1}(0) = \{ (x, y, z) \in \OC^{\times 3} \mid x + y + z = 0 \in \gf^* \} \subset \OC^{\times 3}.
  \end{align}
  Namely $L_{\mathrm{W}}$ consists of triples of points on the sphere that span an equilateral triangle with centre at the origin of the sphere.
  We give a depiction of such a triple/triangle in \Cref{fig:Wigner_su2}.
\end{example}
{\begin{figure}[H]
    \centering
    \includegraphics[width=.4\columnwidth]{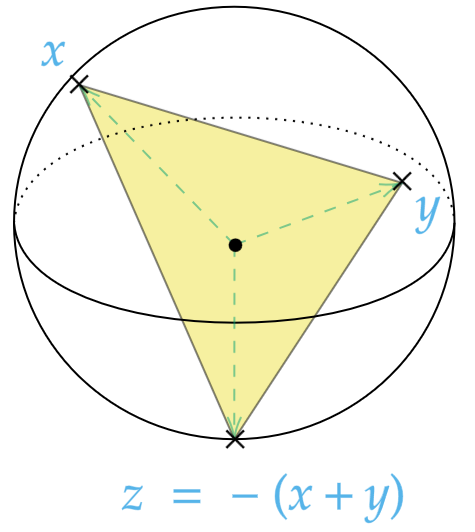}
    \caption{
        A triple of points in $\OC$ spanning an equilateral triangle through the origin.
        The space of such triples/triangles defines $L_{\mathrm{W}}$.
    }
    \label{fig:Wigner_su2}
\end{figure}}
\noindent In this case we can also explicitly write (up to scaling)
\begin{align}
  \ket{L_{\mathrm{W}}} &=
    (z_1 - z_2) (z_2 - z_3) (z_3 - z_1)
    = - z_1^2 z_2 + z_1^2 z_3 - z_2^2 z_3 + z_1 z_2^2 - z_1 z_3^2 + z_2 z_3^2 \\
    &= - \ket{1, 0, -1} + \ket{1, -1, 0} - \ket{-1, 1, 0} + \ket{0, 1, -1} - \ket{0, -1, 1} + \ket{-1, 0, 1}.
\end{align}

The last line is written abstractly in the weight basis of the triple tensor product.
It immediately shows that in real polarisation, where one uses Bohr--Sommerfeld leaves obtained as level-sets of the height function on a sphere, the Lagrangian enforces $m_1 + m_2 + m_3 = 0$, where $m_i \in \{-1, 0, 1\}$ denotes the height of the $i$-th leaf.
\section{A \texorpdfstring{$1$}{1}-dimensional Theory on Graphs}
\label{sec:1D_theory_graphs}
In this section we apply the ideas of the previous text to construct a theory $\tb$ whose partition function on a graph recovers the weight of the corresponding Feynman graph of the theory $\TB$.
We give two presentations of such a construction:

{\bf Hamiltonian/operator formalism presentation:} In \Cref{subsec:t_as_TQM_on_graphs} we construct from a given AKSZ theory $\TB$ a topological quantum mechanics $\tb$ on metric graphs whose amplitudes recover exactly the Feynman weights of $\TB$.

{\bf Lagrangian/path integral presentation:} In \Cref{subsec:PathIntegral_formulation_tb} we give a path integral formulation of the 1d AKSZ theory $\tb$, first on an interval, then on a graph.

As examples, we discuss Chern-Simons and $BF$ theories in different gauges (cf. \Cref{subsubsec:tb_as_TQM_examples}, \Cref{subsubsec:CS_FukayaMorseWitten}).
%
%
%
\subsection{Theory \texorpdfstring{$\tb$}{t} as Topological Quantum Mechanics on Metric Graphs}
\label{subsec:t_as_TQM_on_graphs}
Consider an AKSZ theory $\TB$ on an $m$-manifold $M$ 
(\ref{eq:intro_AKSZ_fields}), (\ref{eq:ST_AKSZ_action}), with space of fields $\FC^\TB = \Omega^\bullet(M) \otimes Y$, with target $Y$ equipped with a degree $m$ Hamiltonian $\Theta_Y$.
Taylor expansion of $\Theta_Y = \sum_{n \geq 3} \frac{1}{n!} c^Y_n$ equips $Y$ with the structure of a cyclic $\LI$-algebra with operations $c^Y_n$.

We associate to $\TB$ a topological quantum mechanics $\tb$ on metric graphs, as in Section \Cref{subsec:HTQFT_graphs}, with the following data:
\begin{itemize}
    \item
        The space of states $\tb$ assigns to a point is $V = \Omega^\bullet(M; Y) = \FC^\TB$ -- it coincides with the space of AKSZ fields of $\TB$.
    \item
        The differential $Q$ on $V$ is the de Rham differential on $M$: $Q = \dR$. 
    %
    \item
        The pairing on $V$ is $\pairing = \int_M \pairing_Y$ -- the BV symplectic form $\omega^\TB$ of theory $\TB$. 
    \item
        The second (gauge-fixing) differential $G$ on $V$ can be chosen in different ways, see examples below.
        We require it to be self-adjoint w.r.t. $\pairing$.\footnote{
            So that the theory $\tb$ is defined on non-oriented graphs.
        }
    \item
        ``Vertex data'': A $\qcLi$-algebra structure on $V$ with operations $c_n^0 = c_n^Y \otimes \int_M - \wedge \cdots \wedge -$ (i.e., $c_n^0(A, \ldots, A) = \int_M c_n^Y(A, \ldots, A)$ with $A \in \Omega^\bullet(M) \otimes Y$ the AKSZ superfield), with $n \geq 3$.
        The operations with nonzero loop number $c^{\geq 1}_\bullet$ are all set to zero.\footnote{
            They are the would-be counterterms, but one does not have counterterms in an AKSZ theory.
        }
        Put another way, $V$ as a $\mr{cL}_\infty$-algebra is the tensor product of the $\mr{cL}_\infty$-algebra $Y$ and the dg Frobenius algebra $\Omega^\bullet(M)$.
\end{itemize}
Then as in \Cref{subsec:HTQFT_graphs} we have pre-amplitudes of $\tb$ on a graph $\Gamma$, 
\begin{align*}
    \PA^\infty_\tb(\Gamma) \in \Omega^\bullet(\RB_+^{\mr{IE}(\Gamma)}) \otimes \mr{Hom}(\ker(H)^{\otimes n},\CB),
\end{align*}
for graphs $\Gamma$ with $n$ leaves (all treated as incoming).
Integrating a pre-amplitude over $\RB_+^{\mr{IE}(\Gamma)}$ we obtain the weight of $\Gamma$ as a Feynman graph in the AKSZ theory $\TB$:
\begin{equation}
    \int_{\RB_+^{\mr{IE}(\Gamma)}} \PA^\infty_\tb(\Gamma)(\underbrace{a,\ldots,a}_n) = F_\Gamma(a),
\end{equation}
with $F_\Gamma(a)$ defined as in (\ref{eq:Feynman_weight_definition}).
Put another way, if $(\PA^\infty_\tb)^k_n$ is the pre-amplitude as a distributional form on the moduli space of metric graphs with $n$ leaves and loop number $k$ (cf. \Cref{rem:Preamplitudes_InvarDef}), we have the amplitude
\begin{equation}
    \frac{1}{n!} (\mathsf{A}_\tb)^k_n = \int_{\mc{MG}^k_n} (\PA^\infty_\tb)^k_n,
\end{equation}
cf. (\ref{amplitude as integral over MG}).
Then for the effective action of theory $\TB$ we have
\begin{equation}
    S^\TB_\mr{eff}(a) = \sum_{n,k\geq 0} \frac{\hbar^k}{n!} (\mathsf{A}_\tb)^k_n(\underbrace{a,\ldots,a}_n). 
\end{equation}
Thus, Feynman diagram expansion of the effective action of theory $\TB$ is recovered in terms of amplitudes of theory $\tb$.
From this viewpoint, the quantum master equation for the effective action $S^\TB_\mr{eff}$ is a consequence of the fact that the amplitudes $\mathsf{A}_\tb$ arrange into a $\qcLi$-algebra structure on $\ker(H)$ (\Cref{theo:HomotopyTransfer_MG}).
\subsubsection{Examples}
\label{subsubsec:tb_as_TQM_examples}
\begin{example}[Chern--Simons Theory]\label{ex:CS_Lorenz_gauge}
    Take $\TB$ to be Chern--Simons theory on a 3-manifold $M$. Here the target is $Y = \gf[1]$, with $\gf$ a Lie algebra with a non-degenerate invariant pairing.
    In this case $\Theta_Y = \frac16 \langle c, [c,c] \rangle$, the cyclic $\LI$ structure on $Y$ has only one operation $c_3^Y = \langle -, [-,-]\rangle$ and Feynman graphs $\Gamma$ have to be trivalent (with possible leaves).
    
    For $G$ one can choose $G_\mr{Lorenz} = \dR^*$ (with corresponding $H = \Delta$ and $\ker(H)$ given by $\gf$-valued harmonic forms), which corresponds to Chern--Simons theory in Lorenz gauge.

    Note that one cannot deform $G$ in this example by a contraction with a vector field as in  (\ref{eq:MorseWitten_G_H}) as this deformation would spoil self-adjointness of $G$.
\end{example}

\begin{example}[\texorpdfstring{$\gl_N$}{gl(N)}-valued Chern--Simons Theory]
\label{ex:glN_CS_in_MQ_gaug_sec51}
    Take $\TB$ to be Chern--Simons theory with structure Lie algebra $\gf = \gl_N$, $N \geq 1$ and fix a collection of functions $f_a$ on $M$ with $a = 1, \ldots, N$, such that differences $f_a - f_b$ are Morse functions for $a\neq b$ with $v_{ab}$ the corresponding gradient vector field. 
    Choose $G$ to be the $N \times N$ matrix of operators on differential forms, with the following entries:
    \begin{equation}
        G_{ab} = \imath_{v_{ab}} + \varepsilon \dR^*.
    \end{equation}
    Then, $H$ is an $N \times N$ matrix of operators with entries $H_{ab} = \Lie{v_{ab}} + \varepsilon \Delta$.
    In this case, for $\varepsilon > 0$, $\ker(H)$ is the space of matrices of twisted harmonic forms (representatives of cohomology of $M$ valued in $\mathfrak{gl}_N$). 
    
    The limit $\varepsilon = 0$ poses certain analytic issues (since $H$ is not a generalised Laplacian in this case).
    With an appropriate regularisation, $\ker(H)$ consists of $N \times N$ matrices with arbitrary differential forms on the diagonal and Morse cochains\footnote{
        Represented by delta-forms supported on unstable manifolds of critical points.
    } of $f_a - f_b$ for off-diagonal entries, $a \neq b$. 
    
    In the limit $\varepsilon = 0$ one has an appealing intersection-theoretic interpretation of the values of Feynman graphs:
    One can interpret the value of a Feynman graph in terms of 
    the count of points of the moduli space of a graph of gradient trajectories, which is cut out by a certain intersection (where the potential issue is transversality of the intersection), see \cite{Fukaya:1996mp}, \cite{Chekeres:2021ieg}.
    For instance, in \cite{Fukaya:1996mp}, it is proven that the value of the Theta-graph with this gauge-fixing (putting three different gradient vector fields on the three edges) is well-defined, i.e., is given by a 
    well-defined (transversal) intersection problem.
    
    Feynman graphs in this example can be thought of as ribbon graphs with borders coloured by labels in $\{ 1,\ldots, N \}$, so that each edge is coloured by a pair of labels $(a, b)$.\footnote{
        The tree sector of this model with $\varepsilon = 0$ was considered in \cite{Chekeres:2021ieg} and gives a repackaging for Fukaya-Morse $A_\infty$ category. The loop enhancement of Fukaya-Morse $A_\infty$-category is work in progress.
    }

    We call the gauge-fixing in this example the Fukaya-Morse-Witten gauge, cf. \cite{Fukaya:1996mp}, \cite{Witten:1982im}.
\end{example}

\begin{example}[\texorpdfstring{$BF$}{BF} Theory]\label{ex:BF_theory_sec51}
    Take $\TB$ to be $BF$ theory on an $m$-manifold $M$.
    In this case the AKSZ target is $Y = \g[1]\oplus \g^*[m-2] \ni (c, \beta)$, with canonical pairing between the summands and with Hamiltonian $\Theta_Y = \frac12 \langle \beta, [c,c] \rangle$.
    
    Put differently, $BF$ theory is Chern--Simons theory with coefficients in a graded Lie algebra $\mathfrak{h} = \gf \oplus \gf^*[m-3]$, with Lie algebra structure given by the semidirect sum of $\gf$ with its (shifted) coadjoint module and quadratic structure given by canonical pairing of $\gf$ with the coadjoint module.
    
    Choose $f$ a Morse function on $M$ with gradient vector field $v$. Then we can choose $G$ to be given by the block matrix of operators on forms
    \begin{equation}\label{eq:G_BF}
        G = \left(
        \begin{array}{cc}
             \imath_{v} + \varepsilon \dR^* & 0  \\
             0                              & -\imath_{v} + \varepsilon \dR^*
        \end{array}
        \right)
    \end{equation}
    with the block structure coming from the decomposition of $Y$ into two summands. The corresponding TQM Hamiltonian is
    \begin{equation}
        H = \left(
        \begin{array}{cc}
              \Lie{v} + \varepsilon \Delta & 0  \\
             0                          & - \Lie{v} + \varepsilon \Delta
        \end{array}
        \right)
    \end{equation}
    Its kernel for $\varepsilon > 0$ consists of pairs of twisted harmonic forms with coefficients in $\gf$ and $\gf^*$, and is isomorphic to cohomology of $M$ with values in $\mathfrak{h}$.
\end{example}
\subsection{Path integral (Lagrangian) description for the theory \texorpdfstring{$\tb$}{t}}
\label{subsec:PathIntegral_formulation_tb}
\subsubsection{Extending the target of theory \texorpdfstring{$\tb$}{t} by a \texorpdfstring{$\cLIsymp$}{cL∞symp}-algebra \texorpdfstring{$\nu$}{ν}}
\label{subsubsec:Extending_target_tb}
Let $(\nu, \rho, \{L_n\})$ be a $\cLIsymp$-algebra (\Cref{def:Min_cyc_Linfty_Symp}) corresponding by geometric quantisation to the target $Y$ of the AKSZ theory $\TB$ with cyclic $\LI$ structure arising from the Taylor expansion of $\Theta_Y = \sum_{n\geq 3} \frac{1}{n!} c_n^Y$.

Consider the 1d AKSZ theory $\tb$ of \Cref{sec:1D_AKSZ} on an interval and extend the target $\NC = \Tan^*(\Tan[1](M \times \RB[1]))$ by the extra factor $\nu$:
\begin{equation}
\label{eq:N_full}
    \NC \lra \NC^\mr{full} = \NC \times \nu.
\end{equation}
I.e., now we consider the theory with space of fields
\begin{equation}
\label{eq:space_of_fields_tb_interval}
    \Map(\Tan[1] I, 
        \underbrace{\Tan^*(\Tan[1](M \times \RB[1])) \times \nu}_{
            \NC^\mr{full}
        }
    ).
\end{equation}
with target Hamiltonian (\ref{eq:Theta_N}) as before (i.e. we choose the trivial Hamiltonian $\Theta_\nu = 0$ for the new factor $\nu$). 
We extend the symplectic form (\ref{eq:omega_N}) by the symplectic form on $\nu$: $\omega_\NC \mapsto \omega_\NC+\omega_\nu$. 

Denote the new AKSZ superfield associated to the factor $\nu$ in the target by $\widetilde{\xi} = \xi + \xi^\dagger \in \Map(\Tan[1] I, \nu)$, with $\xi \colon I \ra \nu$ the 0-form component (a path in $\nu$).
If the prequantum line bundle $\lie{\nu}$ over $\nu$ is trivial, the prequantum connection $\nabla_\nu$ in $\lie{\nu}$ has a global connection 1-form $\alpha_\nu$, and then we extend the action $S^\tb$ as
\begin{equation}
    S^\tb \lra S^\tb + \int_I \xi^* \alpha_\nu.
\end{equation}

However, typically $\lie{\nu}$ is non-trivial (such as in the case of $\nu$ being a coadjoint orbit).
Then, as a generalisation of the formula above, we extend the exponential of the AKSZ action (\ref{eq:Action_tb}) of theory $\tb$ by the parallel transport of $\nabla_\nu$ along the image path of the source interval $I = [0,1]$ in $\nu$:
\begin{equation}\label{extending e^(iS^t) by Hol factor}
    e^{\mr{i} S^\tb} \lra e^{\mr{i} S^\tb} \mr{Hol}_{\nabla_\nu}(\xi(I)) \quad \in\lie{\nu}|_{\xi(1)} \otimes \lie{\nu}^*|_{\xi(0)},
\end{equation}
cf. \cite[Section 7.2]{Mnev:2012qd}.

Let $\NC^\mr{red} = \Tan^* \Tan[1] M \times \nu$ be the AKSZ target (\ref{eq:N_full}) with supergravity factor removed (or: supergravity sector reduced).
We equip $\NC^\mr{red}$ with the structure of a non-minimal $\cLIsymp$-algebra of degree $N = 0$ (cf. \Cref{rem:non_minimal_cLinftysymp_algebra}), with the following data:
\begin{itemize}
    \item
        $\vartheta = p_i \theta^i$ the target Hamiltonian (\ref{eq:Theta_N}) without the supergravity term.
    \item
        $\rho\colon \NC^\mr{red} \rel{} (\overline{\NC^{\mr{red}}})^{[-1]}$ is given by the changing sign of the cotangent coordinates in $\Tan^* \Tan[1]M$ and by $\rho_\nu$ in $\nu$.
    \item
        Lagrangians 
        \begin{equation}\label{eq:Ln_in_Nred}
            \LB_n \coloneqq \CN \mr{Diag}
            \times L_n \quad \subset (\overline{\NC^\mr{red}})^{\times n},
        \end{equation}
        for $n\geq 3$.
        Here, $\mr{Diag}$ is the image of the diagonal embedding $\Tan[1]M \hookrightarrow (\Tan[1]M)^{\times n}$, $\CN$ is the conormal bundle, $L_n \colon \nu^{\times n} \rel{} \mr{pt}^{[m]}$ are the $\cLIsymp$ operations on $\nu$.
\end{itemize}
Geometric quantisation $\QC$ acts on the data of $\NC^\mr{red}$ as follows:
\begin{itemize}
    \item
        $\QC$ maps $\NC^\mr{red}$ to $\Omega^\bullet(M) \otimes Y = \FC^\TB$ – the space of fields of theory $\TB$  (\ref{eq:intro_AKSZ_fields}).
    \item
        $\QC$ maps $\vartheta$ to de Rham operator on $M$.
    \item
        $\QC$ maps $\rho$ to the BV 2-form $\int_M 
        \langle \delta A, \delta A \rangle_Y$ on $\FC^\TB$.
        Here $A\in \Omega^\bullet(M)\otimes Y$ is the AKSZ superfield.
    \item
        $\QC$ maps $\LB_n$ to $\int_M 
        c_n^Y(A, \ldots, A)\colon (\FC^\TB)^{\otimes n} \ra \CB$ – the $n$-th term in the Taylor expansion of the interaction term in the AKSZ action (\ref{eq:ST_AKSZ_action}).\footnote{
            Note that the geometric quantisation of the conormal bundle $\mr{N}^*C \subset \Tan^* X$, for $C \subset X$ a submanifold, is $\delta_C$ – the delta function supported on $C$, as a distribution on $X$.
            By that token the quantisation of $\mr{N}^*\mr{Diag}$ is $\delta_\mr{Diag} \colon \alpha_1 \otimes \cdots \otimes \alpha_n \mapsto \int_{(\Tan[1]M)^{\times n}} \mu^\mr{can}_{\Tan[1] M^{\times n}} (\prod_i \pi_i^* \alpha_i) \cdot \delta_\mr{Diag} = \int_{\Tan[1] M} \mu^\mr{can}_{\Tan[1]M} \alpha_1 \cdots \alpha_n = \int_M \alpha_1 \wedge \cdots \wedge \alpha_n$ for $\alpha_1, \ldots, \alpha_n \in \Ci(\Tan[1]M) \cong \Omega^\bullet(M)$.
            Here $\mu^\mr{can}_{\Tan[1]M}$ is the canonical Berezinian on $\Tan[1]M$, such that the integral of $\alpha \in \Ci(\Tan[1]M)$ against it is the integral of (the top component of) $\alpha$ interpreted as a form on $M$ over $M$.
            By this mechanism, the term $\mr{N}^*\mr{Diag}$ in (\ref{eq:Ln_in_Nred}) results in the integral over $M$ of a polynomial expression in the superfield evaluated on the diagonal.
        }
\end{itemize}
As an additional gauge-fixing datum, we fix a function $G^\mr{cl} \in \Ci(\NC^\mr{red})_{-1}$ quantising to the gauge-fixing differential $G$ on $\FC^\TB$.
\subsubsection{Path Integral of Theory \texorpdfstring{$\tb$}{t} on an Interval (an edge of a graph)}
\label{subsubsec:PathIntegral_formulation_tb_Interval}
In order to construct the path integral of $\tb$ on a graph $\Gamma$, we first consider edges of $\Gamma$:
On an interval $I$, $\tb$ is a 1d AKSZ theory we studied in \Cref{sec:1D_AKSZ}, with space of fields 
\begin{align}
\label{eq:Fields_tb_on_Interval}
    \FC^\tb_I = \mr{Map}(\Tan[1]I, \underbrace{\Tan^* \Tan[1]M \times \nu \times \Tan^* \Tan[1] \RB[1]}_{\NC^\mr{full}}).
\end{align}
As in \Cref{subsec:Simplification_SUGRA}, we integrate out the 1d supergravity fields with boundary conditions $c|_{\partial I} = \phi|_{\partial I} = 0$;\footnote{
    These boundary conditions on $c, \phi$ correspond to the reduction of the SUGRA sector of the space of states, cf. \Cref{subsubsec:Reduction_SUGRA_Boundary}.
}
the result is a function on the space of bulk SUGRA zero-modes 
$\YC'_\SUGRA = \Tan^*[-1] \Tan[1] \RB = \{(e^I, \epsilon^I, p_c^I, p_\phi^I)\}$.

We then use the gauge-fixing determined by the fermion $\Psi = - \int_I \dR t \, e^I G^\mr{cl}$ (\ref{eq:Psi}) to integrate out the remaining fields as in (\ref{eq:Partition_function_tb_G}) and (\ref{eq:Partitionfunction_coupled_AKSZintegrated}), eventually reproducing the propagator of theory $\TB$:
\begin{align*}
    K^\TB = \int_{\Tan[1]\RB_+} \dR e^I \dR \epsilon^I \, \underbrace{e^{- e^I H - \epsilon^I G}}_{Z^\tb_I(e^I,\epsilon^I|p_c^I = p_\phi^I = 0)}\quad \in \mr{End}_{-1}(\FC^\TB).
\end{align*}
This the integral of the partition function of $\tb$ on the interval, seen as a function of SUGRA residual fields, integrated over a Lagrangian 
\begin{equation}
\label{eq:L_SUGRA}
    \LC^\SUGRA = \mr{N}^*(\Tan[1] \RB_+) \subset \YC'_\SUGRA
\end{equation}
given by $e^I > 0$, $\epsilon^I$ any, $p_c^I = p_\varphi^I = 0$.\footnote{
    Note that we are naturally led to the integral over a chain \emph{with boundary} in $\Tan[1]\RB$. 
    This yields boundary contributions from the points where $e^I = 0$, the zero-length limit of an edge.
    These boundary terms contribute to the quadratic structure equations of the effective theory (see the proof of \Cref{theo:HomotopyTransfer_MG}).
    An example of this for the Fukaya--Morse category enhanced by intersection conditions is discussed in \cite[Section 4.2]{Chekeres:2021ieg}.
}
Under the identification $e^I = T$, $\epsilon^I = \dR T$, the BV integral over SUGRA residual fields is the same as the integral over moduli of metric intervals in TQM, in the formalism of \Cref{subsec:HTQFT_graphs}.

For external edges (``leaves'') we want to choose a different gauge-fixing Lagrangian for the SUGRA sector:
\begin{equation}
    \LC^\SUGRA_\mr{external} = \{(e^I = + \infty, \, \epsilon^I \text{ any}, p_c^I \text{ any}, \; p_\phi^I = 0)\} = \mr{N}^* (\{ + \infty \} \times \RB[1])
\end{equation}
Then we have
\begin{equation}
    \int_{\LC^\SUGRA_\mr{external}} e^{- e^I H - \epsilon^I G + \mr{i} \epsilon^I p_c^I } = P_{\ker(H)},
\end{equation}
the projector onto $\ker(H)$.
\subsubsection{Path integral of theory \texorpdfstring{$\tb$}{t} on a graph}
\label{subsubsec:Path_integral_tb_graph}
On a graph $\Gamma$, for the space of fields of theory $\tb$ we take the fibre product of the spaces \eqref{eq:Fields_tb_on_Interval} for edges over Lagrangians $\LB_n$ \eqref{eq:Ln_in_Nred} for vertices, with $n \in \NB$ the valence of the vertex (we also set supergravity fields $c, \phi$ to zero at vertices).
Thus, Lagrangians $\LB_n$ provide the sewing conditions for fields of theory $\tb$ at a vertex.
I.e., we have
\begin{equation}
    \FC^\tb_\Gamma = \{\psi \in \Map(\Tan[1]\Gamma, \NC^\mr{full}) \mid \forall v \in \mr{V}(\Gamma), \; \{\psi_e(v)\}_{e \; \text{edge incident to } v} \in \LB^\mr{full}_{\mr{val}(v)}
    \}.
\end{equation}
Here $\psi = (\til{x}^i, \til\theta^i, \til{p}_i, \til\pi_i; \til\xi; \til{c}, \til\phi, \til{p}_c, \til{p}_\phi)$ is the collective notation for all superfields and $\LB_n^\mr{full} \coloneqq \LB_n \times \mr{N}^* (\{0\}\subset (\Tan[1]\RB[1])^{\times n})$.
The last factor of $\LB_n^\mr{full}$ imposes the condition that $c,\phi$ vanish at vertices.

On leaves we impose asymptotic/boundary conditions, fixing the values of $x^i,\theta^i, \xi_\mc{B}$ (and setting $c,\phi$ to zero), where we imagine that $\nu$ is equipped with a real polarisation with $\xi_\mc{B}$ the coordinate on the space of Lagrangian leaves.

The action of $\tb$ on $\Gamma$ is the sum of actions on edges: $S^\tb_\Gamma = \sum_{e \in \mr{E}(\Gamma)} S^\tb_e$.

Integrating out the fields while retaining the SUGRA residual fields for each internal edge, restricting to the Lagrangian 
\begin{equation}
    \LC^\SUGRA_\mr{internal} =
    \mr{N}^* \Big(
        (\Tan[1] \RB_+)^{\mr{IE}(\Gamma)}
    \Big),
\end{equation}
we obtain the 
pre-amplitude $\PA^\infty_\tb(\Gamma)(a)$:\footnote{
    Note that technically we are recovering the preamplitudes of the effective theory on the full space of fields, i.e. as a function of $A$ instead of $a$.
    However, following \Cref{rem:Weights_full_vs_residual} we neglect this difference for the sake of simplicity and work with residual fields instead.
}
%
\begin{equation}
    \PA_\tb^\infty(\Gamma)(a) = \int_{\prod_\mr{leaves} \LC_\mr{external}^\SUGRA} 
    \int_{\LC_\sigma^0\subset \FC_\sigma} e^{\mr{i} \Phi^*((S^\tb_\Gamma)_\sigma+(S^\tb_\Gamma)^\mr{eff}_\mr{SUGRA})}
    \Big|_{\LC^\SUGRA_\mr{internal}}
\end{equation}
with $\Phi$ the symplectomorphism as in Section \ref{subsubsec:Reproduce_prop}.
Then, integrating over $\LC^\SUGRA_\mr{internal}$, we obtain the value $F_\Gamma(a)$ of $\Gamma$ as a Feynman graph in theory $\TB$.

Summing over connected graphs, we recover the effective action of theory $\TB$:
\begin{equation}
    S^\TB_\mr{eff}(a) = \sum_\Gamma \frac{\hbar^{l(\Gamma)}}{|\Aut(\Gamma)|} \int_{\LC^\SUGRA_\mr{internal}} \PA_\tb^\infty(\Gamma)(a).
\end{equation}
Here $l(\Gamma)$ is the loop number and the integral is over the moduli space of geometric data (via BV-pushforward).

In particular if $\ker(\Hhat) \simeq \mr{H}^\bullet(M)$, the tree-level part of this expansion (order $\hbar^0$) recovers a representation of the Massey operations on $M$.
\begin{remark}
  The factor $\frac{1}{| \Aut(\Gamma) |}$ stems from global graph automorphisms.
  It appears when passing from the ``stacky'' integral $\int_{\Geom(\Gamma) / \sim}$, where we mod by local diffeomorphisms (diffeomorphisms of an edge relative to its endpoints) and global diffeomorphisms (graph automorphisms), to a sum over graphs.
  
  One might also ask where the weight system $\hbar^{l(\Gamma)}$ stems from.
  One explanation is a scaling ambiguity of the partition function of the $1$-dimensional theory.
  It leads to the weight system when compatibility with gluing/cutting is imposed.
  This is a toy version of a similar effect for $2D$ Yang--Mills \cite[Section 6]{Cordes:1994fc}.
\end{remark}

\subsubsection{Examples: Chern--Simons and \texorpdfstring{$BF$}{BF} theory via coadjoint orbits, \texorpdfstring{$\mathfrak{gl}(N)$}{gl(N)}-Chern--Simons in Fukaya-Morse-Witten gauge
}
\label{subsubsec:CS_FukayaMorseWitten}
\begin{example}[Chern--Simons Theory via Coadjoint Orbits]\!\!\!\!\!\footnote{
    This construction is partially informed by the process of adding Wilson line observables as an auxiliary theory coupled to ambient theory, cf. \cite{Alekseev:2012wc}, \cite{Mnev:2012qd}.
}
    Returning to \Cref{ex:CS_Lorenz_gauge}, consider as $\TB$ the Chern--Simons theory on a 3-manifold $M$ with coefficients in a simple Lie algebra $\gf$.
    Assume that $\gf$, viewed as the adjoint module over itself, is the quantisation (by Kirillov's orbit method) of a coadjoint orbit $\OC \subset \gf^*$ and assume that $L_\mr{W} \subset \overline{\OC}^{\times 2}\times \OC$ is the Wigner Lagrangian quantising to the Lie bracket $[-,-]$, cf. \Cref{subsec:Lie_bracket_from_Lagrangian}.

    Then we construct theory $\tb$ with $\nu = \OC^{[1]}$ the coadjoint orbit, with structure maps $\rho \colon \xi \mapsto - \xi$, with $L_3 = \rho^{-1} \circ L_\mr{W} \subset \overline{\OC}^{\times 3}$ the ``cyclic version'' of the Wigner Lagrangian above and with $L_{\neq 3} = 0$.

    For the gauge-fixing datum, we set
    \begin{equation}
        G^\mr{cl} = -\mr{i} \left(g^{ij}(x) p_i \pi_j -  g^{ij}(x) \Gamma^l_{jk} \pi_l \pi_i \theta^k\right),    
    \end{equation}
    cf. (\ref{eq:Gcl_MorseWitten}).
    Note that this is a ghost degree $-1$ function on $\NC^\mr{red} = \Tan^* \Tan[1] M \times \OC^{[1]}$ independent of the last factor.
    Its quantisation recovers the gauge-fixing operator $G = \dR^*$ of Chern--Simons theory in Lorenz gauge (acting as identity in the second factor in $\FC^\TB = \Omega^\bullet(M) \otimes \gf[1]$).

    As a sub-example, one can take $\gf = \su(2)$, with $\OC$  corresponding to the adjoint representation and $L_3 = L_\mr{W}$ of \Cref{example: L_W for su(2)}.
\end{example}

\begin{remark}[Twisting by a local system]
    One can recover the perturbative expansion of the path integral of Chern--Simons theory around a nonzero flat connection on $M$ given by a connection 1-form $A_0\in \Omega^1(M; \gf)$.
    In order to do that (see \cite[Section 7.2]{Mnev:2012qd}), one can deform the target AKSZ Hamiltonian by an extra term, $\Theta_\NC \ra \Theta_\NC + \Theta_\OC$, with
    \begin{equation}
        \Theta_\OC = \langle A_0, \mu_\OC(\xi) \rangle = (A_0)^a_i(x) \theta^i \mu_\OC(\xi)_a \quad \in \Ci(\Tan[1]M \times \nu) \subset \Ci(\NC^\mr{red}).
    \end{equation}
    Here $\mu_\OC \colon \OC \hookrightarrow \gf^*$ is the moment map.
    Adding this term has the effect of deforming the de Rham differential $\dR_M$ (generated by $p_i \theta^i$) to $\dR_{A_0} = \dR_M + \mr{ad}_{A_0}$.

    Note that flatness of $A_0$ is equivalent to the requirement that $\{\Theta_\NC + \Theta_\OC, \Theta_\NC + \Theta_\OC\} = 0$ (i.e. that the deformed Hamiltonian still generates a cohomological vector field on $\NC$).
\end{remark}

\begin{remark}[Wilson Graphs]
    One can consider putting different integrable coadjoint orbits $\OC_e$ (not necessarily the one quantising to the adjoint module) on different edges $e \in \mr{E}(\Gamma)$ and putting appropriate Lagrangians (quantising to intertwiners) as sewing conditions in vertices.
    This corresponds to evaluating the leading quasiclassical contribution to the expectation value of a Wilson graph in Chern--Simons theory, see e.g. \cite{Reshetikhin:2010zz}. 
\end{remark}

\begin{example}[BF Theory via Coadjoint Orbits]
\label{example: BF via coad orbits}
    Revisiting \Cref{ex:BF_theory_sec51}, consider $BF$ theory on an $m$-manifold $M$ with coefficients in a simple Lie algebra $\gf$.
    Assume that $\gf$ is the quantisation of a coadjoint orbit $\OC$ with Wigner Lagrangian $L_\mr{W} \subset \overline{\OC} \times \overline{\OC} \times \OC$.
    
    Then we set
    %
    \begin{align*}
        \nu = \underbrace{\OC^{[1]}}_{\nu_A} \sqcup \underbrace{\overline{\OC}^{[m-2]}}_{\nu_B},
    \end{align*}
    with $\rho$ the map $(\xi_A, \xi_B) \mapsto (\xi_B, \xi_A)$ and the structure Lagrangian
    \begin{align*}
        L_3 = L_\mr{W} \subset \overline{\nu}_A^{\times 2} \times \overline{\nu}_B,
    \end{align*}
    cf. Example \ref{example: dequantization of g+g^*}.

    Let us denote the r.h.s. of (\ref{eq:Gcl_MorseWitten}) for a gradient vector field $v$ by $G^\mr{cl,MW}_v$.
    Using the splitting 
    \begin{equation}
    \label{eq:Nred_BF}
        \NC^\mr{red} = (\Tan^* \Tan[1] M \times \nu_A) \sqcup (\Tan^* \Tan[1] M \times \nu_B),   
    \end{equation}
    we choose the gauge-fixing datum $G^\mr{cl}$ to be given by $G^\mr{cl}_A = G^{\mr{cl, MW}}_v
    $ on the first component in the r.h.s. of (\ref{eq:Nred_BF}) and by
    $G^\mr{cl}_B= G^{\mr{cl, MW}}_{-v}
    $ on the second component.
    Here $v$ is the gradient vector field for a Morse function $f$, as in \Cref{ex:BF_theory_sec51}.
    The quantisation of this gauge-fixing datum yields the gauge-fixing operator (\ref{eq:G_BF}).
\end{example}

\begin{example}[Chern--Simons in Fukaya--Morse--Witten Gauge]
    Revisiting \Cref{ex:glN_CS_in_MQ_gaug_sec51}, consider Chern--Simons theory on a 3-manifold $M$ with structure group $\gf = \gl_N 
    $ for some $N \geq 2$.
    We set $\nu = \sqcup_{a, b = 1}^N \pt_{ab}^{[1]}$ – a collection of $N^2$ points, equipped with the data $\rho, L_3$ as in \Cref{ex:dequantisation_of_glN}:
    \begin{align*}
        \rho \colon \pt_{ab} \mapsto \pt_{ba}, \qquad L_3 = \sum_{a, b, c} (\pt_{ab} \times \pt_{bc} \times \pt_{ca}-\pt_{ba} \times \pt_{cb} \times \pt_{ac}).
    \end{align*}
    In this case the target of theory $\tb$ is $\NC^{\mr{full}} = \bigsqcup_{a,b} \Tan^* \Tan[1](M \times \RB[1]) \times \pt_{ab}$ is a disjoint union of $N^2$ copies of $\Tan^* \Tan[1](M \times \RB[1])$.
    Thus, a field of theory $\tb$ on an edge consists of a choice of a bi-label $(a, b)$ (thought of as decorations of the two sides of a ribbon – a thickening of the edge – by $a,b \in \{1, \ldots, N\}$) and element of $\Map(\Tan[1] I, \Tan^* \Tan[1](M \times \RB[1]))$.

    As in \Cref{ex:glN_CS_in_MQ_gaug_sec51}, we fix a collection of functions $f_a$ on $M$, with $a = 1, \ldots, N$, such that differences $f_a - f_b$ are Morse-Smale for $a \neq b$.
    We denote $v_{ab}$ the gradient vector field of $f_a - f_b$.

    For the gauge-fixing datum we choose the function $G^\mr{cl}$ given on $\Tan^* \Tan[1] M \times \pt_{ab}$ by $G^\mr{cl}_{ab} = G^{\mr{cl, MW}}_{v_{ab}}$.
    It quantises to $G$ as in \Cref{ex:glN_CS_in_MQ_gaug_sec51}.
\end{example}

\appendix
%
\section{A Berezin integral formula for the Hodge star}
Denote 
\begin{equation}
\sigma_k\colon= (-1)^{\frac{k(k-1)}{2}},
\end{equation}
so that for $\theta^1,\ldots,\theta^k$ odd variables, one has 
$\theta^k\cdots \theta^1 = \sigma_k \theta^1\cdots \theta^k$.
\begin{proposition}
\label{prop:HodgeStar_T1M}
    Let $(M, g)$ be a Riemannian manifold of dimension $n$ and let $f(x, \theta) \in \Fun(\Tan[1]M)$ be a $p$-form on $M$, expressed locally as a smooth function of local coordinates $x^i, \theta^i$ on $\Tan[1]M$, of polynomial degree $p$ in $\theta^i$.
    %
    Then the Hodge star $*_g$ acting on $f(x, \theta)$ is given by the following Berezin integral:
    \begin{align}
    \label{eq:HodgeStar_T1M}
        *_g f(x, \theta) &= 
        \sigma_{n-p}
        \int_{\Tan_x[-1] M \ni \overline{\theta}} (\det(g_x))^{- \frac12} D\overline{\theta}^n \cdots D\overline{\theta}^1 \, e^{g_x(\overline{\theta}, \theta)} f(x, \overline{\theta})
    \end{align}
\begin{remark}
    Let $\mr{dvol}_x=(\det g_x)^{\frac12}dx^1\wedge\cdots \wedge dx^n\in \wedge^n T^*_x M$ be the Riemannian volume element. Then the Berezinian in the integral in (\ref{eq:HodgeStar_T1M}) can be invariantly written as ${(\mr{dvol}_x)^{-1}\in \wedge^n T_x M=\mr{Ber}_x(M)}$. Here we are exploiting the fact that Berezinian line bundle is dual to the line bundle of top forms on $M$.
\end{remark}
    \begin{proof}
    Operators acting on $f$ in both sides in (\ref{eq:HodgeStar_T1M}) are (a) maps of $C^\infty(M)$ modules and (b) are independent of the choice of coordinates on $\Tan[1]M$. Thus, we can assume that the local coordinates are such that the metric at $x$ is $g_x=\sum_{i=1}^n (\mr{d}x^i)^2$ and it suffices to check (\ref{eq:HodgeStar_T1M}) for $f=\theta^1\cdots \theta^p$. In this case, the r.h.s. of (\ref{eq:HodgeStar_T1M}) is
    \begin{equation*}
    \begin{aligned}
        &\sigma_{n-p} \int_{\RB^n[-1]} D\bar\theta^n\cdots D\bar\theta^1\, \underbrace{e^{\sum_{i=1}^n \bar\theta^i \theta^i}}_{\prod_{i=1}^n (1+\bar\theta^i\theta^i)}  \cdot \bar\theta^1\cdots\bar\theta^p
        \\
        &=\sigma_{n-p} \int_{\RB^n[-1]} \overleftarrow{\prod_{i=1}^n} \Big(D\bar\theta^i (1+\bar\theta^i\theta^i) \Big)\cdot
        \bar\theta^1\cdots\bar\theta^p \\
        &= \sigma_{n-p} \int_{\RB^n[-1]} \overleftarrow{\prod_{i=p+1}^n} (D\bar\theta^i\cdot  \bar\theta^i\theta^i )\cdot \overleftarrow{\prod_{i=1}^p} (D\bar\theta^i \cdot \bar\theta^i)\\
        &= \sigma_{n-p}\theta^n\cdots \theta^{p+1}=\theta^{p+1}\cdots \theta^n
        =*_g(\theta^1\cdots \theta^p)
    \end{aligned}
    \end{equation*}
    -- the l.h.s. of (\ref{eq:HodgeStar_T1M}).
\end{proof}
\end{proposition}
\begin{remark}
    The Hodge star can be factorised as a composition of maps
    \begin{equation}\label{Hodge star factorized through OFT}
        *_g\colon\quad \Omega^p(M)\xrightarrow{\mr{OFT}} \mathcal{V}^{n-p}(M)  \xrightarrow{g} \Omega^{n-p}(M),
    \end{equation}
    where the first map is the odd Fourier transform (cf. (18) in \cite{schwarz1993geometry}) based on the Riemannian volume element on $M$, mapping polyvectors to differential forms;
    the second map is induced by the isomorphism between tangent and cotangent bundle given by the metric.
    Since the odd Fourier transform is a Berezin integral, (\ref{Hodge star factorized through OFT}) is an equivalent formulation of (\ref{eq:HodgeStar_T1M}).
\end{remark}
\section{A formula for the codifferential in local coordinates}
\begin{proposition}\label{prop: d^* in loc coordinates}
    The codifferential $\dR^*$ associated to the Hodge star $*_g$ can be expressed in terms of $x^i, \theta^i$, the metric $g$ and its Christoffel symbols $\Gamma$ as:
    \begin{align}
    \label{eq:Codifferential_T1M}
        \dR^* &= - g^{jk} \dell{}{\theta^j} \left( \dell{}{x^k} - \Gamma^l_{km} \theta^m \dell{}{\theta^l} \right) .
    \end{align}
\begin{proof}
    Let – as before – $f(x, \theta)$ be a $p$-form on $M$, expressed locally as a smooth function on $\Tan[1]M$.
    Using (\ref{eq:HodgeStar_T1M}) we obtain the following expression for $\dR *_g f(x, \theta)$, where in the following we drop the subscript $x$ from the metric:
    \begin{align}
        &~ \dR *_g f(x, \theta) \\
        =&~ \sigma_{n-p} \left( \theta^k \dell{}{x^k} \right) \int_{\Tan_x[-1] M \ni \overline{\theta}} (\det(g))^{- \frac12} D\overline{\theta} \, e^{g(\overline{\theta}, \theta)} f(x, \overline{\theta}) \\
        =&~ \sigma_{n-p} (-1)^n \int_{\Tan_x[-1] M} (\det(g))^{- \frac12} D\overline{\theta} \, \underbrace{e^{g(\overline{\theta}, \theta)} \theta^k}_{g^{jk} \dell{}{\overline{\theta}^j} e^{g(\overline{\theta}, \theta)}} \left(
            - \frac12 \underbrace{\partial_k \ln(\det(g))}_{g^{lm} (\partial_k g_{lm})}
            + \partial_k g(\overline{\theta}, \theta)
            + \dell{}{x^k}
        \right) f(x, \overline{\theta}) \\
        =&~ \sigma_{n-p} (-1)^{n+1} \int_{\Tan_x[-1] M} (\det(g))^{- \frac12} D\overline{\theta} \, e^{g(\overline{\theta}, \theta)} \underbrace{g^{jk} \dell{}{\overline{\theta}^j} \left(
            - \frac12 g^{lm} (\partial_k g_{lm})
            + \partial_k g(\overline{\theta}, \theta)
            + \dell{}{x^k}
        \right)}_{\bigstar} f(x, \overline{\theta}) .\label{B.5}
    \end{align}
    Above we integrate by parts to move the derivative in $\overline{\theta}$.
    In the following, we will simplify the expression $\star$ – with integration as above assumed – by using integration by parts several times to exchange $\theta^i$ for $- g^{il} \dell{}{\overline{\theta}^l}$:
    \begin{align}
        \bigstar \colon \quad&~ g^{jk} \dell{}{\overline{\theta}^j} \left(
            - \frac12 g^{lm} (\partial_k g_{lm})
            + \partial_k g(\overline{\theta}, \theta)
            + \dell{}{x^k}
        \right) \\
        =&~ g^{jk} \left(
            \dell{^2}{x^k \partial \overline{\theta}^j}
            - \frac12 g^{lm} (\partial_k g_{lm}) \dell{}{\overline{\theta}^j}
            + (\partial_k g_{ji}) \theta^i
            \underbrace{- (\partial_k g_{lm} \theta^l \overline{\theta}^m)}_{\mathclap{- (\partial_k g_{lm}) g^{lr} \overline{\theta}^m \dell{}{\overline{\theta}^r} + g^{lm} (\partial_k g_{lm})}} \dell{}{\overline{\theta}^j}
        \right) \\
        =&~ g^{jk} \left(
            \dell{^2}{x^k \partial \overline{\theta}^j}
            + \frac12 g^{lm} (\partial_k g_{lm}) \dell{}{\overline{\theta}^j}
            - g^{il} (\partial_k g_{ji}) \dell{}{\overline{\theta}^l}
            - \underbrace{(\partial_k g_{lm}) g^{lr} \overline{\theta}^m \dell{}{\overline{\theta}^r} \dell{}{\overline{\theta}^j}}_{\Gamma^r_{km} \overline{\theta}^m \dell{}{\overline{\theta}^r} \dell{}{\overline{\theta}^j} }
        \right) \\
        =&~ 
            g^{jk} \dell{^2}{x^k \partial \overline{\theta}^j}
            - g^{jk} \Gamma^r_{km} \overline{\theta}^m \dell{}{\overline{\theta}^r} \dell{}{\overline{\theta}^j}
            - g^{km} \Gamma^j_{km} \dell{}{\overline{\theta}^j}
        \\
        =&~ g^{jk} \dell{}{\overline{\theta}^j} \left( \dell{}{x^k} - \Gamma^l_{km} \overline{\theta}^m \dell{}{\overline{\theta}^l} \right) .
    \end{align}
    The sign in (\ref{B.5}) is $\sigma_{n-p}(-1)^{n+1} = \sigma_{n-(p-1)}(-1)^{p+1}$.
    Since $\dR^*$ – acting on $p$-forms – is defined as $\dR^* \coloneqq (-1)^p *_g^{-1} \dR *_g$, we apply (\ref{eq:HodgeStar_T1M}) again to obtain (\ref{eq:Codifferential_T1M}).
\end{proof}
\end{proposition}
\section{Dimension Counts for Conjecture \ref{conj: O^3//G=pt}}
\label{Appendix: dimension counts for Conjecture}
    As a sanity check of Conjecture \ref{conj: O^3//G=pt}, consider the case $\gf = \so(n)$, $n \geq 5$.\footnote{
        We are excluding the case $\so(4)\simeq \su(2)\oplus \su(2)$ which is not a simple Lie algebra and the case $\so(3) \simeq \su(2)$ which is discussed in Example \ref{example: L_W for su(2)} separately.
    }
    The coadjoint orbit $\OC$ corresponding to the adjoint representation is the orbit through the point $\mr{diag}(1, 1, \underbrace{0, 0, \ldots, 0}_{n-4},-1,-1) \in \mathfrak{h}^*$ and so can be identified with the homogeneous space $\SO(n) / (\SO(n-4) \times \U(2))$.
    Thus, it has dimension $\dim\OC = \frac{n(n-1)}{2}-\frac{(n-4)(n-5)}{2} - 4 = 4n - 14$.
    In particular, it is not an orbit of generic type (which has a higher dimension $\frac{n(n-1)}{2}-\left[\frac{n}{2}\right]$).
    Then, the virtual dimension of the symplectic reduction is
    \begin{equation}
    \label{dim^vir(O^3//G)}
        \dim^\textrm{vir}\mu^{-1}(0)/G = 3\dim \OC - 2\dim\gf
    \end{equation}
    it vanishes for $n = 6, 7$ and is negative for $n \neq 6, 7$ (which means that for $n \neq 6, 7$, zero is not a regular value of the moment map (\ref{moment map}) and the action of $G$ on $\mu^{-1}(0)$ is not locally free -- has a stabiliser).\footnote{
        Note that if for $x\in \mu^{-1}(0)$ the codimension of the image of $\dR \mu|_x$ in $\gf^*$ is $k \geq 0$, then the stabiliser of $x$ under the action of $G$ is a $k$-dimensional subgroup  $H_x \subset G$.
        Then the tangent space to the symplectic reduction is $T_x (\mu^{-1}(0) / G) = (\ker \dR \mu|_x) \,/\, (\gf / \mr{Lie}(H_x))$ and has dimension $(3 \dim \OC - \dim\gf + k) - (\dim\g - k) = 3 \dim \OC - 2 \dim \gf + 2k$.
        So, the actual dimension of the symplectic reduction is $\dim \mu^{-1}(0) / G = (\dim^\mr{vir} \mu^{-1}(0) / G) + 2k$.
    }
    If the virtual dimension (\ref{dim^vir(O^3//G)}) were to attain a positive value, it would indicate a contradiction with Conjecture \ref{conj: O^3//G=pt}, but fortunately we are not seeing positive values.

    Similarly, for $\gf = \spin(2n)$, the orbit $\OC$ passes through the point $\mr{diag}(1, \underbrace{0,\ldots,0}_{n-1}, -1, \underbrace{0,\ldots,0}_{n-1}) \in \mathfrak{h}^*$ and can be identified with the homogeneous space $\USp(2n)/(\USp(2n-2)\times \U(1))$, and thus $\dim \OC = n(2n+1) - (n-1)(2n-1) - 1 = 4n - 2$.
    The virtual dimension (\ref{dim^vir(O^3//G)}) is negative for all $n \geq 2$.
    
    We summarise the dimension counts in a table and include the case $\gf=\su(n)$ for comparison (where $\OC = \SU(n)/(\U(n-2)\times \U(1))$ is the orbit through $\mr{diag}(1,\underbrace{0, \ldots,0}_{n-2}, -1) \in \mathfrak{h}^*$, and so $\dim\OC = n^2 - 1 - (n-2)^2 - 1 = 4n - 6$).

    \begin{center}
    \begin{tabular}{c|c|c|c|c}
        $\gf$ & $\dim\gf$ & $\dim\OC$ & $3\dim\OC-2\dim\gf$ & $\dim \OC_\mr{generic}$\\ 
        &&& $=\dim^\mr{vir}\mu^{-1}(0)/G$ &  $=\dim\gf-\mr{rank}\,\gf$
        \\ \hline
         $\su(2)$ &  $3$ & $2$ & $0$ &  $2$ \\
         $\su(3)$ &  $8$ & $6$ & $2$ &  $6$ \\
         $\su(4)$ &  $15$ & $10$ & $0$ &  $12$ \\
         $\su(5)$ &  $24$ & $14$ & $-6$ &  $20$ \\
         $\vdots$ & $\vdots$ & $\vdots$& $\vdots$& $\vdots$ \\
         \thinline
         $\so(5)$ &  $10$ & $6$ & $-2$ &  $8$ \\
         $\so(6)\simeq \su(4)$ &  $15$ & $10$ & $0$ &  $12$ \\
         $\so(7)$ &  $21$ & $14$ & $0$ &  $18$ \\
         $\so(8)$ &  $28$ & $18$ & $-2$ &  $24$ \\
         $\vdots$ & $\vdots$ & $\vdots$& $\vdots$& $\vdots$\\
         \thinline
         $\spin(2)\simeq \su(2)$ &  $3$ & $2$ & $0$ &  $2$ \\
         $\spin(4) \simeq \so(5)$ &  $10$ & $6$ & $-2$ &  $8$ \\
         $\spin(6)$ &  $21$ & $10$ & $-12$ &  $18$ \\
         $\spin(8)$ &  $36$ & $14$ & $-30$ &  $32$ \\
         $\vdots$ & $\vdots$ & $\vdots$& $\vdots$& $\vdots$
    \end{tabular}
    \end{center}
    The fact that $\dim^\mr{vir} > 0$ for $\gf = \su(3)$ should not worry us, since this case is not covered by Conjecture \ref{conj: O^3//G=pt}.
    The cases of $\su(2)$ and $\su(3)$ are the only ones where $\OC$ is the generic orbit.
\clearpage
\phantomsection
\addcontentsline{toc}{section}{References}
\printbibliography
\end{document}